\documentclass[a4paper,11pt]{article}
\usepackage[margin=2.5cm]{geometry}
\usepackage[]{algorithm2e}
\usepackage{amsmath,amsthm,amssymb}
\usepackage{tikz-cd,graphicx}
\usepackage{float}
\usepackage{indentfirst}
\usepackage{mathtools}
\usepackage{hyperref}
\usepackage{subcaption}
\usepackage{cleveref}

\newcommand{\E}{\mathbb{E}}

\newcommand{\Z}{\mathbb{Z}}

\newcommand{\str}{\mathit{stretch}}
\newcommand{\spread}{\mathit{spread}}
\newcommand{\FCB}{\mathit{FCB}}

\newcommand{\DP}{\mathsf{DP}}

\usepackage{aliascnt}
\newtheorem{theorem}{Theorem}[section]
\newaliascnt{lemma}{theorem}
\newtheorem{lemma}[lemma]{Lemma}
\aliascntresetthe{lemma}

\newaliascnt{corollary}{theorem}
\newtheorem{corollary}[corollary]{Corollary}
\aliascntresetthe{corollary}

\author{ 
Glencora Borradaile\thanks{School of Electrical Engineering and Computer Science, Oregon State University.}
\and
Erin Wolf Chambers\thanks{Department of Computer Science, Saint Louis University.}
\and
David Eppstein\thanks{Department of Computer Science, University of California, Irvine.}
\and
William Maxwell\footnotemark[1]
\and
Amir Nayyeri\footnotemark[1]
}

\title{Low-stretch spanning trees of graphs with bounded width}\date{}

\begin{document}

\maketitle
\begin{abstract}
We study the problem of low-stretch spanning trees in graphs of bounded width: bandwidth, cutwidth, and treewidth.      
We show that any simple connected graph $G$ with a linear arrangement of bandwidth $b$ can be embedded
into a distribution $\mathcal T$ of spanning trees such that the expected stretch
of each edge of $G$ is $O(b^2)$. Our proof implies a linear time algorithm for
sampling from $\mathcal T$.  Therefore, we have a linear time algorithm that finds a
spanning tree of $G$ with average stretch $O(b^2)$ with high probability. We also describe a deterministic linear-time algorithm for computing a spanning tree of $G$ with average stretch $O(b^3)$.
For graphs of cutwidth $c$, we construct a spanning tree with stretch $O(c^2)$ in linear time.
Finally, when $G$ has treewidth $k$ we provide a dynamic programming algorithm computing a minimum stretch spanning tree of $G$ that runs in polynomial time with respect to the number of vertices of $G$.
\end{abstract}

\section{Introduction}
Let $G=(V,E)$ be an unweighted, connected graph with $m$ edges and $n$ vertices, and $T$ be any spanning tree of $G$. For any $(u, v) \in E$, the stretch of $(u, v)$ with respect to $T$ is $\str_T(u, v) = d_T (u, v)$, where
$d_T(u, v)$ denotes the length of the unique $u$-to-$v$ path in $T$. The stretch of $T$ is then defined to be $\str(T) = \frac{1}{m}\sum_{(u,v)\in E}{\str_T(u, v)}$. 

As minimal distance preserving structures, low-stretch spanning trees are a fundamental concept that have been studied extensively;
they have also found  applications in computer science in problems such as the $k$-server problem~\cite{Alon95GraphTheoGame}, minimum cost communication trees~\cite{Peleg98MinCommST}, and solving diagonally dominant linear systems \cite{MillerSDD}. 
Perhaps the first notable structural result is the paper by Alon et al.~\cite{Alon95GraphTheoGame}, where they show that any general graph has a spanning tree of stretch $O(\exp(\sqrt{\log n \log \log n}))$ and that there exist graphs with minimum stretch  $\Omega(\log n)$.  A series of papers~\cite{Elkin08LowStretchST, Abraham08NearTightStretchST, Koutis11KMPSolver, Abraham12PetalLowStretST} followed the result of Alon et al., culminating in the recent construction of Abraham and Neiman of an $O(\log n \log \log n)$ stretch spanning tree for general graphs, which is almost tight considering the  $\Omega(\log n)$ lower bound.  The existence of spanning trees with bounded average distortion is often implied by a stronger statement that the graph can be embedded into a distribution of spanning trees such that the expected stretch of any edge is bounded. 

Given these results for general graphs, a natural question is to consider restricted classes of graphs, both in terms of finding better bounds than general graphs for some classes of graphs, as well as finding lower bounds that match the general case in others.
 For example, we know that constant factor stretch spanning trees exist for $k$-outerplanar graphs: they have stretch $c^k$ for a constant $c$~\cite{Gupta04CutsTreesL1Emb, Emek11KOuterLowStretchST}.
On the lower bound side, we also know that grid graphs, which are planar, have a lower bound of $\Omega(\log n)$ on their stretch, so we cannot hope to get constant factor for this class.
Additionally, Gupta et al.~\cite{Gupta04CutsTreesL1Emb} found a family of bounded treewidth graphs (in fact, series parallel graphs) whose minimum stretch spanning trees have stretch  $\Omega(\log n)$. 

In light of these bounds, the search for families of graphs that might have smaller stretch must be limited to classes of graphs that exclude these examples.
In this regard, a natural and still-open question is whether bounded pathwidth graphs admit a spanning tree of sublogarithmic stretch.
In fact, we conjecture that bounded pathwidth graphs admit constant stretch spanning trees.  In this paper, we make progress towards this conjecture by showing this is true for bounded bandwidth (\autoref{thm:expected-stretch}) and bounded cutwidth graphs (\autoref{cw}); both classes are contained within the family of bounded pathwidth graphs.
More precisely, we prove:  
\begin{itemize}
\item For every $n$-vertex graph of bandwidth $b$ there exists a random distribution over spanning trees of the graph,
such that the expected stretch of any individual edge of the graph is $O(b^2)$. The random distribution can be sampled in linear time given a bandwidth-$b$ linear arrangement of the graph, or constructed explicitly in quadratic time.  
\item Under the same assumptions, a spanning tree $T$ of average stretch $O(b^3)$ can be constructed deterministically in linear time.
\item Every $n$-vertex graph of cutwidth $c$ has a spanning tree $T$ of average stretch $O(c^2)$. $T$ can be constructed from a cutwidth-$c$ linear arrangement of the graph in linear expected time.
\item We provide a dynamic programming algorithm computing the minimum stretch spanning tree of an unweighted graph with treewidth $k$. Our algorithm runs in $O(2^{3k} k^{2k} n^{k+1})$ time.
\end{itemize}
It is important to note that our algorithms require either a linear arrangement or a tree decomposition realizing the width as input, and computing such structures is NP-hard \cite{Papadmimitriou1976,Arnborg1987,Gavril1977}.

Lee and Sidiropoulos~\cite{Lee13PathwidthTreeRE} show that a bounded pathwidth graph admits an embedding into a distribution of trees with constant distortion.  In this paper, we conjecture that a similar result holds for embedding into a distribution of \emph{spanning} trees.  For embedding of bounded bandwidth graphs into normed spaces see Carrol et al.~\cite{Carroll06EmbBandwidthL1} and Bartal et al.~\cite{Bartal13BandwdithLowDimEmb}.

The key insight by which we obtain these results lies in the connection between spanning trees of low-stretch and 
 fundamental cycle bases of low weight.
Any spanning tree $T$ of $G$ naturally gives a fundamental cycle basis for $G$:
for each $e=(u, v) \in E\backslash T$, the basis contains the unique cycle in $T\cup \{e\}$.
The weight of this basis is defined to be the sum of the lengths of its cycles.
A graph $G$ has a spanning tree of average stretch $O(\log n)$ if and only if it has a fundamental cycle basis of weight $O(m\log n)$.
Similarly, a cycle basis of length $O(m)$ is equivalent to a spanning tree of stretch $O(1)$.
(The relationship between $T$'s stretch and fundamental cycle basis will be discussed in more detail in the next section.)

Shortest fundamental cycle bases have been studied as a basic structure of graphs and for their different applications in graph drawing~\cite{Fruchterman91GraphDraw}, electrical engineering~\cite{Bollobas98ModernGraphTheo}, chemistry~\cite{Gleiss01PhD}, traffic light planning~\cite{Kohler04MinTotalDelay}, periodic railway time tabling,~\cite{Liebchen07PriodicTimetable, Reich14PhD}, and kinematic analysis of mechanical structures~\cite{Cassell_1974}.

\section{Preliminaries}

\subsection{Cycle bases}

Given a simple, connected, unweighted graph $G$ with $n$ vertices and $m$ edges the {\em cycle space} of $G$ is an $m - n + 1$ dimensional vector space over $\Z_2$ that spans the cycles in $G$.
In this context a cycle in $G$ is any subgraph of $G$ with even degree.
We call a basis of this vector space a {\em cycle basis}, and the weight of a cycle basis is the sum of the lengths of the cycles in the basis.
Given a spanning tree $T$ of $G$ we call a cycle formed by adding a non-tree edge to $T$ a {\em fundamental cycle} with respect to $T$.
Every spanning tree $T$ of $G$ yields a basis of the cycle space using the fundamental cycles induced by the $m - n + 1$ edges in $G \setminus T$.
We call a basis of this form a {\em fundamental cycle basis}.
Each cycle in the fundamental cycle basis created by $T$ corresponds to exactly one edge in $G \setminus T$. We call this edge the {\em fundamental edge} of the cycle.

\subsection{Fundamental cycle bases and low-stretch spanning trees}

The weight of a fundamental cycle basis with respect to a tree $T$ is closely related to the stretch of $T$.
The {\em stretch} of an edge $e=(u, v)$ in $G$ with respect to $T$, denoted $\str_T(e)$,  is defined as the length of the unique $u$-to-$v$ path in $T$.
The stretch of $T$ is defined as the mean stretch of the edges, \[\str(T) = \frac{1}{m} \sum_{e \in E(G)} \str_T(e).\]
Let $\FCB(T)$ denote the weight of the fundamental cycle basis corresponding to $T$.
By observing that the length of a fundamental cycle induced by an edge $e$ is $\str_T(e) + 1$ we see that the fundamental cycle basis with respect to $T$ is related to the stretch of $T$ by 
\begin{align}
\label{eqn:FCBvsStretch}
\FCB(T) = m \cdot \str(T) + m - 2n + 2
\end{align}
It follows that $\FCB(T) = O(m)$ if and only if $\str(T) = O(1)$.

\subsection{Linear arrangements}
A bijective map $\phi \colon V(G) \rightarrow \{1,2,...,n\}$ is called a {\em linear arrangement} of $G$.
For any subset of vertices $S \subseteq V(G)$ if $s \in S$ maximizes $\phi$ restricted to $S$ we call it the {\em right endpoint} of $S$; similarly if $s$ minimizes $\phi$ restricted to $S$ we call it the {\em left endpoint} of $S$.
If $u$ and $v$ are the left and right endpoints of $S$ we define the {\em spread} of $S$ to be $\phi(v) - \phi(u)$.
For any vertex $v$ we call the sets $\{ u \in V(G) \mid \phi(u) < \phi(v)\}$ and $\{ u \in V(G) \mid  \phi(v) < \phi(u) \}$ the \emph{left and right sides} of $v$, respectively. 

\subsection{The arrangement tree}

Given a linear arrangement $\phi$ of $G$ the {\em arrangement tree} $A$ is defined as a balanced binary tree with the following two properties.
The leaves of $A$ are in bijection with $V(G)$ and each internal node $v$ is mapped to the subgraph of $G$ induced by the vertices corresponding to the descendent leaves of $v$. More specifically we construct $A$ as follows:
let $n$ be the number of vertices in $G$, and let $p$ be the largest power of two that is less than $n$.
Let the left subtree of $A$ be constructed recursively from the first $p$ vertices in the linear arrangement,
and let the right subtree be constructed recursively from the remaining $n-p$ vertices (\autoref{fig:arrangement-tree}).

\begin{figure}[!htb]
\centering\includegraphics[width=0.4\textwidth]{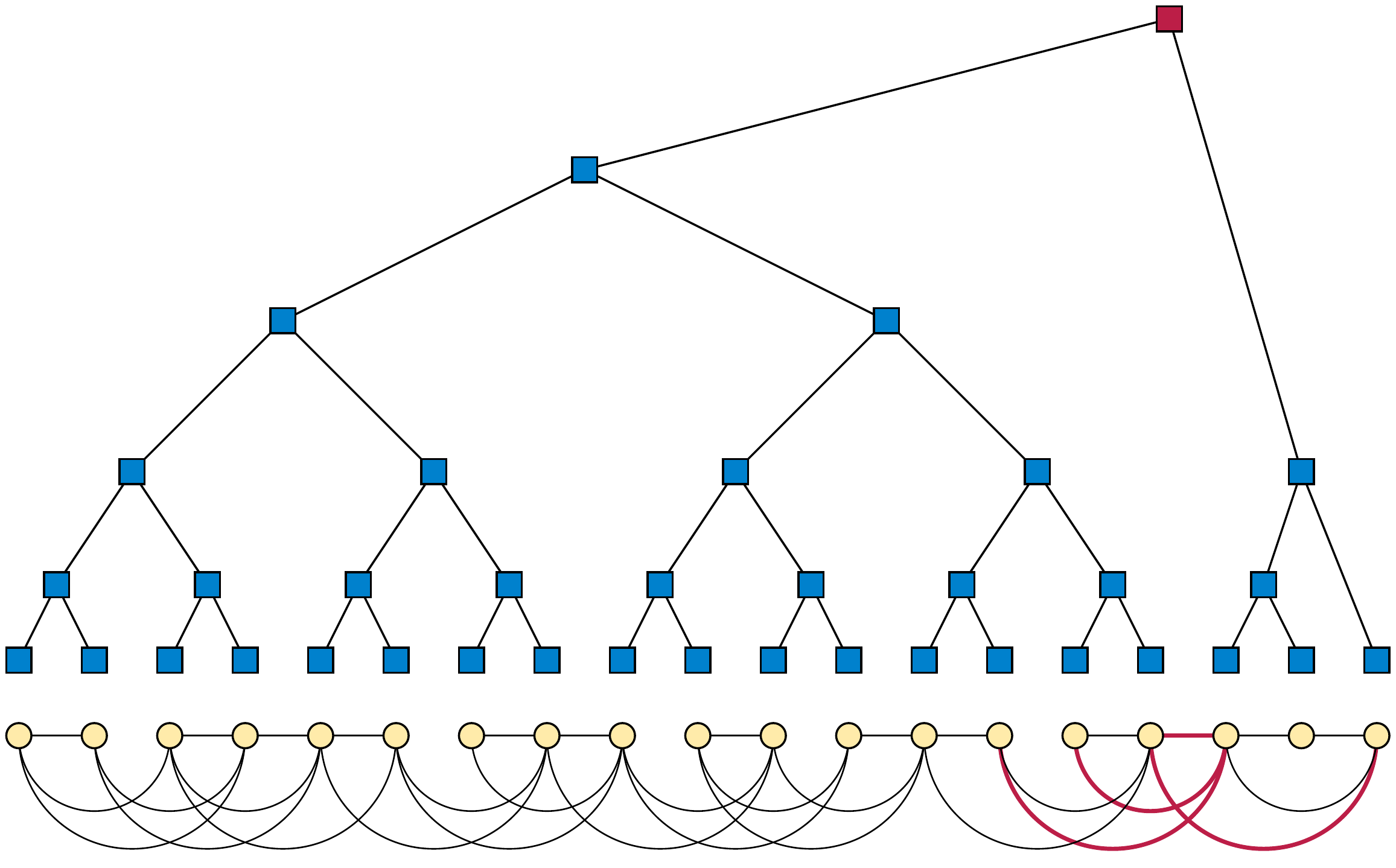}
\caption{Linear arrangement of a graph of bandwidth three and its arrangement tree. The root node and the edges split by the root node are marked in red.}
\label{fig:arrangement-tree}
\end{figure}

We denote the induced subgraph of the leaves descending from $v$ by $G_v$.
Consider the children $x$ and $y$ of $v$ in $A$. The induced subgraph $G_v$ has the form $G_v = G_x \cup G_y \cup S_v$ where $S_v$ is the set of edges connecting $G_y$ and $G_x$.  We call $S_v$ the set of edges {\em split} by $v$. Note that each edge is split by exactly one vertex.

\subsection{Bandwidth and cutwidth}
The {\em bandwidth} of a linear arrangement $\phi$ of a graph $G$ is defined as 
\[ \max_{(u, v) \in E(G)} |\phi(u) - \phi(v)|. \]
Note that $| \phi(u) - \phi(v) |$ is the spread of $(u, v)$ with respect to the arrangement tree arising from $\phi$.
The bandwidth of $G$ is the minimum bandwidth over all possible linear arrangements.
In a graph with bandwidth $b$ we have $\deg(v) \leq 2b$ for all $v \in V(G)$. Hence, when $b = O(1)$ we have $|E(G)| = O(n)$.
Consider the induced subgraph $G_v = G_x \cup G_y \cup S_v$ corresponding to node $v$ of $A$ with $x$ as the left and $y$ as the right child of $v$.
Any edge $(q, r) \in S_v$ with $q$ and $r$ from $G_x$ and $G_y$, respectively, has spread at most $b$.  So, if $r$ is $i$ positions away from the left endpoint of $G_y$ then $q$ is at most $b-i$ positions away from the right endpoint of $G_x$.
It follows that
\begin{equation}
\label{edge-bound}
|S_v| \leq \frac{1}{2}(b-1)(b-2) = O(b^2).
\end{equation}

The {\em cutwidth} of a linear arrangement $\phi$ of a graph $G$ is defined as \[\max_{i \in \Z} | \{ (u, v) \in E(G) \mid \phi(u) \leq i, \; \phi(v) \geq i+1 \}|.\]
The cutwidth of $G$ is the minimum cutwidth over all linear arrangements.
The cutwidth measures the number of edges that cross a fixed position in the linear arrangement.

\subsection{Tree decompositions}
A \textit{tree decomposition} of a graph $G$ is a tree $\mathcal{D} = (\mathcal{I}, \mathcal{E})$ where the vertex set $\mathcal{I}$ is in bijection with a collection $\{B_i\}_{i \in \mathcal{I}}$ of subsets of $V(G)$, called \textit{bags}, meeting the following conditions.

\begin{enumerate}
\item Every vertex $v \in V(G)$ is contained in some bag. That is, $\bigcup_{i \in \mathcal{I}} B_i = V(G)$.

\item For every edge $(u, v) \in E(G)$ there exists an $i \in \mathcal{I}$ with $u, v \in B_i$.

\item For all $v \in V(G)$ the subgraph induced by the set of bags containing $v$ is a tree.
\end{enumerate}
The \textit{width} of a tree decomposition is defined to be $\max_{i \in \mathcal{I}} |B_i| - 1$. The \textit{treewidth} of $G$ is the minimum width over all of its tree decompositions, denoted $k$.
We will use the notation $D(B)$ to refer to the set of vertices in the bag $B$ and the descendants of $B$. Similarly, by $A(B)$ we denote the set of vertices in $B$ and the ancestors of $B$.
We call a tree decomposition a \textit{nice tree decomposition} if it meets the following extra conditions.
\begin{enumerate}
\setcounter{enumi}{3}
\item $\mathcal{D}$ is a rooted binary tree.
\item If $i, j, k \in \mathcal{I}$ with $j$ and $k$ the children of $i$, then $B_i = B_j = B_k$.
\item If $j$ is the child of $i$ and $\deg(i) = 2$ then either $B_j \subset B_i$ and $|B_i| = |B_j|+1$ or $B_i \subset B_j$ and $|B_i| = |B_j|-1$.
\end{enumerate}
We call the parent bags satisfying property 5 \textit{join nodes}.
We call the parent bags satisfying the two conditions of property 6 \textit{introduce nodes} and \textit{forget nodes}, respectively.
Without loss of generality we may assume all tree decompositions are nice since any tree decomposition can be transformed into a nice tree decomposition in polynomial time \cite{nice}. Further, we may assume that every leaf bag contains only one vertex and the root bag is a forget node containing only one vertex.
\section{Spanning trees from linear arrangements}\label{sec:st}

Both our construction of a random family of spanning trees with low expected stretch on each edge and our construction of a deterministic spanning tree with low mean stretch will depend on a construction of spanning trees from arrangement trees, which we now describe.

Although we will use a different construction algorithm, our tree can be described as the one constructed by the following greedy algorithm:

\begin{algorithm}
  Given a graph $G$ and arrangement tree $A$:\\
  $T \leftarrow \emptyset$ \\
  for node $x \in A$ in leaf-to-root order:\\
  \qquad for edge $e \in S_x$ in increasing order by spread:\\
  \qquad \qquad if $T \cup \{e\}$ is acyclic, add $e$ to $T$. \\
  Return $T$.
    \caption{Spanning tree from a linear arrangement}\label{alg:st}
\end{algorithm}

This algorithm is simply Kruskal's algorithm for the minimum spanning tree of $G$, with each edge weighted by the height in the arrangement tree of the least common ancestor of the edge endpoints with ties broken by spread.
Because the result is a minimum spanning tree for these edge weights, we can construct the same tree by any other minimum spanning tree algorithm.
Finding the lowest common ancestor for all edges in $G$ can be done in $O(n)$ time~\cite{linearLCA}.
The algorithm of Fredman and Willard~\cite{int-mst}, which finds a minimum spanning tree  of a graph with integer weights in $O(n)$ time, implies that our algorithm can be implemented in linear time.

\begin{lemma}
\label{lem:mst-edge-stretch}
Let $e$ be an arbitrary edge of $G$, let $i$ and $j$ (with $1\le i<j\le n$) be the positions of the endpoints of $e$ in the linear arrangement, and let $p$ be the largest power of two that divides an integer in the half-open interval $[i,j)$.
Then the stretch of $e$ in the tree constructed as above is $O(p)$.
\end{lemma}

\begin{proof}
Let $v$ be the node of the arrangement tree with $e\in S_v$. From our construction of the arrangement tree it follows that the number of leaf descendants of $v$ is at least $p+1$ and at most $2p$.
By the greedy algorithm for the construction of a spanning tree, the spanning tree contains a path connecting the endpoints of $e$ within these at most $2p$ descendants, for otherwise $e$ itself would have been added to the spanning tree. Therefore, the stretch of $e$ is at most $2p-1$.
\end{proof}

\section{Embedding into a distribution of trees}

Let $G$ be any graph having a linear arrangement $\phi$ of bandwidth $b$. In this section, we construct a random distribution over spanning trees $\mathcal T$ of $G$ with the property that each edge of $G$ has expected stretch $O(b^2)$.
That is, for an arbitrary edge $e$ (chosen independently from the construction of $\mathcal T$) we have $\mathbb{E}_{\mathcal T}[\str(e)] = O(b^2)$.
A single tree from the distribution can be sampled in time $O(n)$, and the entire distribution can be constructed explicitly in time $O(n^2)$.

Let $n$ be the number of vertices in $G$,
and let $n'$ be the smallest power of two greater than or equal to $2n$ (so, $n' = \Theta(n)$).
Let $G'$ be formed from $G$ by adding $n' - n$ isolated vertices.
Consider the $n'-n \ge n$ different linear arrangements $\phi_i$ of $G'$ obtained from the linear arrangement $\phi$ of $G$ by placing $i$ isolated vertices before the vertices of $G$ and $n' - n -i$ vertices after the vertices of $G$ (for $0\leq i\leq n'-n$).
Denote the collection of arrangement trees of these linear arrangements by $\mathcal{A} = \{A_i\}_{i=1}^{n' - n}$.  For each arrangement tree $A_i \in \mathcal{A}$, Algorithm~\ref{alg:st} produces a tree $T_i$.  Our random distribution $\mathcal T$ is generated by choosing $i$ uniformly at random and, based on that choice, selecting tree $T_i$.

Given a fixed choice of edge $e$, define $\ell(A_i)$ to be the node $v$ of the arrangement tree $A_i$
such that $e\in S_v$ (that is, the endpoints of $e$ are in distinct children of $v$). Given two arrangement trees $A_i$ and $A_j$ we say $A_i \equiv A_j$ if the rightmost leaf descendants of the left children of $\ell(A_i)$ and $\ell(A_j)$ are equal.
That is, $A_i \equiv A_j$ are equivalent if and only if $e$ is split in the same position of the linear arrangements $\phi_i$ and $\phi_j$.
Note that $\equiv$ is an equivalence relation that is defined with respect to a fixed $e$.

Therefore, we can calculate the expected spread of $e$ by concentrating only on a single equivalence class $[A]$ of $\equiv$. Since the bound holds for every equivalence class, the same expected spread will hold for our entire random distribution, by averaging over the equivalence classes.

Given an arrangement tree $A_i$ (chosen from a fixed equivalence class $[A]$) let $v_i$ be the node of $A_i$ such that $e\in S_{v_i}$ (that is, $v_i$ splits $e$), and let $h_i$ be the height of $v_i$ in the arrangement tree.
Then for all $A_j$ in the same equivalence class with $h_i=h_j$, we have $G_{v_i}=G_{v_j}$ and the edges in this induced subgraph have the same minimum spanning tree weights, so they also have
$T_i\cap G_{v_i}=T_j\cap G_{v_j}$.
Within these two subtrees these nodes have the same two paths connecting the endpoints of $e$. Because this path depends only on the height $h_i$ and not on $i$ itself, we denote it $P_{h_i}$.
Different heights may have the same associated paths.
We say that $h_i$ is a \emph{critical height} if $P_{h_i}\ne P_{h_i-1}$; that is, if $h_i$ is the lowest height that gives rise to its path.

\begin{lemma}\label{lem:ncritical}
For a fixed choice of edge $e$ and equivalence class $[A]$ there are $O(b)$ critical heights.
\end{lemma}
\begin{proof}
Let $A_i,A_j \in [A]$ be arrangement trees that split the edge $e$ at vertices $v_i$ and $v_j$, respectively. Further, we assume $v_i$ and $v_j$ are at heights $h_i$ and $h_j = h_i - 1$ where $h_i$ is a critical height.
We denote the spanning trees produced by Algorithm~\ref{alg:st} with input $A_i$ and $A_j$ by $T_i$ and $T_j$.
The associated induced subgraphs are related by the inclusion $G_{v_j} \subset G_{v_i}$.

We now describe the ways in which $T_{v_i} = T_i \cap G_{v_i}$ can differ from $T_{v_j} = T_j \cap G_{v_j}$.
By the construction of the equivalence relation every edge split by $v_j$ is also split by $v_i$, that is $S_{v_j} \subset S_{v_i}$.
The edges in $T_{v_j} \setminus S_{v_j}$ must be included in $T_{v_i}$ since their weights are the same in both $A_i$ and $A_j$. This is because in the linear arrangement $G_{v_i}$ adds an equal number of vertices to the left and right of $G_{v_j}$ and this number is equal to a power of two.
It follows that $T_{v_i}$ differs from $T_{v_j}$ by the addition of non-split edges, the potential addition of split edges, and the potential removal of split edges.

Consider the case when there exists some edge $e' \in T_{v_j} \cap S_{v_j}$ but $e' \notin T_{v_i}$.
The edge $e'$ was added to $T_{v_j}$ by Algorithm~\ref{alg:st} because it connected to previously disconnected components of $G_{v_j}$.
These connected components must have already been contained in a larger connected component of $G_{v_i}$, since otherwise Algorithm~\ref{alg:st} would have picked $e'$ for $T_{v_i}$.
It follows that these connected components must have been connected by the addition of a non-split edge not contained in $T_{v_j}$.

When $e' \in T_{v_i} \cap S_{v_i}$ but not in $T_{v_j}$ then $e'$ must contain an endpoint outside of $G_{v_j}$.
Since there are $O(b)$ vertices within $b$ positions away from the split point, and once a critical height excludes a split edge it cannot be reintroduced to the spanning tree, we see that at most $O(b)$ split edges can be added across all critical heights.

The height $h_i$ can only be a critical height if $T_{v_i} \cap G_{v_j}$ differs from $T_{v_j}$, otherwise $P_{h_i} = P_{h_j}$.
Hence, $h_i$ can only be a critical height if $T_{v_i}$ excludes a split edge appearing in $T_{v_j}$.
The number of split edges at the smallest critical height is $O(b)$ because these edges form an acyclic subgraph on the $O(b)$ vertices within $b$ positions away from the split point.
Since an edge can be excluded from the spanning tree at a critical height at most once we conclude that there are $O(b)$ critical heights.
\end{proof}

\begin{theorem}
\label{thm:expected-stretch}
For an arbitrary edge $e$  (chosen independently from the construction of $\mathcal T$)
the expected stretch of $e$ is $O(b^2)$.
\end{theorem}
\begin{proof}
Let $[A]$ be any equivalence class of the equivalence relation $\equiv$, and let $\str_{[A]}(e)$ denote the expected stretch of $e$ over all arrangement trees from the class $[A]$.  Also, let $h_1 < h_2 < \ldots < h_k$ be the critical heights of $e$ in $[A]$.  Finally, let $v$ be the (random) vertex in the arrangement tree that splits $e$, and let $H_v$ be the random variable of $v$'s height.
We have
\[
\E[\str_{[A]}(e)] = \sum_{i=1}^{k}{len(P_{h_i})\cdot Pr[P_{h_i}]},
\]
where $Pr[P_{h_i}]$ is the probability that $P_{h_i}$ is the path connecting the endpoints of $e$ in the (randomly) selected tree.  It follows, by the definition of critical heights, that
\[
Pr[P_{h_i}] = Pr[h_i \leq H_v < h_{i+1}] \leq Pr[h_i \leq H_v] = O(\spread(e)/2^{h_i}).
\]
In addition, we have $len(P_{h_i}) = O(2^{h_i})$ by \autoref{lem:mst-edge-stretch}.
Putting everything together, we have 
\[
\E[\str_{[A]}(e)] = \sum_{i=1}^{k}{O(2^{h_i})\cdot O(\spread(e)/2^{h_i})} = O(k\cdot \spread(e)) = O(b^2) ,
\]
as $k = O(b)$ by \autoref{lem:ncritical}, and $\spread(e)\le b$ by the definition of bandwidth.

Since $\str(e)$ is a weighted average of $\str_{[A]}(e)$ for different classes $[A]$, and $\str_{[A]}(e) = O(b^2)$ for all classes $[A]$, we conclude that $\str(e) = O(b^2)$.
\end{proof}

\section{Bounded cutwidth}

A theorem from Chung \cite{fan} says that for any graph $G$ with cutwidth $c$ there exists a subdivision of $G$ with bandwidth $c$. However, this entails expanding the number of edges by a factor of $c$, so combining this with our construction of low-stretch spanning trees for low-bandwidth graphs would give us a tree with average stretch $O(c^3)$. In this section we provide a direct construction that obtains stretch $O(c^2)$.
The proof of \autoref{cw} is almost identical to that of \autoref{thm:expected-stretch}, however since we do not have the inequality $\spread(e) \leq c$ we instead compute the expected stretch of the tree rather than the expected stretch of a single edge.
\begin{theorem}\label{cw}
A graph $G$ with cutwidth $c$ has a spanning tree with expected stretch $O(c^2)$.
\end{theorem}
\begin{proof}
We apply the same construction for a random distribution of spanning trees as in \autoref{thm:expected-stretch} to a linear arrangement of $G$ with cutwidth $c$.
We show that the expected stretch a spanning tree produced by Algorithm~\ref{alg:st} on a randomly chosen arrangement tree from the distribution is $O(c^2)$. Therefore, there exists a spanning tree with stretch at least as good as this expected value.

As before, we fix an equivalence class of arrangement trees $[A]$ from our random distribution. Let $h_1< h_2 < \dots<h_k$ denote the critical heights of $[A]$.
Since at $h_k$ there are at most $O(c)$ split edges, we can conclude that there are at most $O(c)$ critical heights. 
As in the proof of \autoref{thm:expected-stretch}, for a fixed edge $e$ the expected stretch is given by \[\mathbb{E}[\str_{[A]}(e)] = \sum_{i=1}^k len(P_{h_i}) \cdot Pr[P_{h_i}].\]
We have that $Pr[P_{h_i}] = O(\spread(e)/2^{h_i})$ and $len(P_{h_i}) = O(2^{h_i})$, hence $\mathbb{E}[\str_{[A]}(e)] = O(c \cdot \spread(e))$.
Let $T$ be the spanning tree constructed by Algorithm~\ref{alg:st} from the randomly selected arrangement tree. We compute the expected stretch of $T$ by
\[\mathbb{E}[\str(T)] = \frac{1}{m} \sum_{e \in E(G)} O(c \cdot \spread(e)). \]
Note that $\sum_{e \in E(G)} \spread(e) \leq cn$ since by the definition of cutwidth at most $c$ edges cross any given interval in the linear arrangement.
Hence, $\mathbb{E}[\str(T)] = O(c^2)$.

\end{proof}

\begin{corollary}\label{cw-cor}
Any graph with cutwidth $c$ has a fundamental cycle basis with weight $O(c^2n)$.
\end{corollary}

Because this method produces high expected stretch for edges of high spread,
it is not clear how to strengthen this result to obtain a distribution with low-stretch for each edge,
as we did for bandwidth. We leave the question of whether this is possible as open for future research.

\section{Deterministic low-stretch spanning tree}
In this section we show that given a graph $G$ and a linear arrangement $\phi$ with bandwidth $b$ Algorithm~\ref{alg:st} produces a spanning tree of $G$ with stretch $O(b^3)$. As stated in Section~\ref{sec:st} such a spanning tree can be constructed in linear time.
Throughout this section we denote the spanning tree produced by Algorithm~\ref{alg:st} by $T$, and we denote the arrangement tree arising from $\phi$ by $A$.

We use a charging scheme to pay for the cycles in the fundamental basis created by our spanning tree algorithm.
Each fundamental cycle with sufficiently large spread will be assigned a charge. Moreover, the sum of the charges is an upper bound on the sum of the lengths of the cycles.


We are now ready to define the key component to our charging scheme.
Let $x$ be a node in $A$. A {\em long component} of $G_x$ is a connected component of $G_x$ that includes at least one vertex within distance $b$ of each endpoint in the linear arrangement of $G_x$. The number of long components in $G_x$ will be denoted with $\ell_x$.

\begin{lemma}\label{long-component}
For any $x \in V(G)$ we have (1) $\ell_x \leq b$ and (2) if $x$ is the parent of $y$ then $\ell_x \leq \ell_y$.
\end{lemma}
\begin{proof}
Since a long component is a special type of connected component a vertex can be in at most one long component.
A long component contains at least one vertex from the first $b$ vertices in the linear arrangement. This implies there can be at most $b$ long components, hence $\ell_x \leq b$.

Let $x$ be a node in $A$ with left child $y$ and right child $z$.
Recall that $G_x = G_y \cup G_z \cup S_x$.
The left endpoint of $G_x$ is the left endpoint of $G_y$, and the right endpoint of $G_x$ is the right endpoint of $G_z$.
Any edge in $S_x$ connects a vertex within the rightmost $b$ vertices of $G_y$ to a vertex within the leftmost $b$ vertices of $G_z$.
Therefore a long component in $G_x$ must contain a long component in $G_y$, a long component in $G_z$, and an edge in $S_x$, thus $\ell_x \leq \ell_y$.
\end{proof}

Here we introduce a charging scheme that will be used to pay for the cycles added to our basis.
For any node $x$ in the arrangement tree $A$ let $n_x$ be the number of leaf descendants of $x$.
If $x$ is the parent of $y$ and $z$ such that $\ell_x < \ell_y$ and $\ell_x < \ell_z$ we assign a charge $c_x = n_y + n_z$ to $x$, if $\ell_x < \ell_y$ and $\ell_x = \ell_z$ we assign a charge $c_x = n_y$ to $x$, similarly if $\ell_z < \ell_x$ and $\ell_y = \ell_z$ we assign $c_x = n_z$, otherwise $c_x = 0$. 
Next, we show that the sum over all charges is $O(n)$.

\begin{lemma}\label{linear-charge}
The sum of the charges is linear in the number of vertices in $G$. That is, $\sum_{x \in V(A)} c_x \leq bn$.
\end{lemma}
\begin{proof}
Consider the set $J$ of nodes with exactly $j$ long components and non-zero charge.
If $u, v \in J$ such that $v$ is a descendent of $u$, then all nodes on the $u$ to $v$ path are in $J$ since by \autoref{long-component} the number of long components is monotonic in depth.
Let $x$ be a node on this path, let $z$ be its child on the path, and let $y$ be its child off the path.  
If both $x$ and $y$ have $j$ long components our charging scheme makes $c_y = 0$, therefore $y \notin J$ and $c_x$ is the number of leaf descendants of $y$. 
Therefore, the sets of leaf descendants from which every node in $J$ derives its charge are disjoint.
Thus, $\sum_{x \in J} c_x \leq n$.
By \autoref{long-component} the number of long components in any induced subgraph is at most $b$, therefore $\sum_{x \in V(A)} c_x \leq bn$.
\end{proof}

Recall that the spread of a fundamental cycle $C$ is defined to be $\phi(v_\ell) - \phi(v_r)$ where $v_\ell$ and $v_r$ are the left and right endpoints of $C$.
In \autoref{spread-factor} we show that the spread of $C$ is within a constant factor of its length.
In \autoref{charge-cycle} we show that $C$'s fundamental edge induces a charge that is within a constant factor of the spread of $C$.
This justifies the use of our charging scheme.

\begin{lemma}\label{spread-factor}
If $C$ is a cycle with length $|C|$ and spread $s$, then we have the inequality $\frac{2s}{b} \leq |C| \leq s + 1$.
\end{lemma}
\begin{proof}
The upper bound is trivial.
Conversely, decompose $C$ into the two unique $v_\ell$-to-$v_r$ paths.
Each edge in these paths has a spread of at most $b$ in the linear arrangement, so each path needs at least $\frac{s}{b}$ edges.
Therefore, $\frac{2s}{b} \leq |C|$.
\end{proof}

Let $C$ be a fundamental cycle of $T$ with length $|C|$, spread $s \geq 4b$, and whose fundamental edge is in $S_x$.
Since $C$'s fundamental edge is in $S_x$, $G_x$ must be the first induced subgraph in the leaf-to-root ordering that contains $C$ since every tree edge of $C$ must be added to $T$ before the fundamental edge is considered by Algorithm~\ref{alg:st}.
Let $S = \{v \in V(G) \mid \phi(v_\ell) \leq \phi(v) \leq \phi(v_r)\}$ where $v_\ell$ and $v_r$ are the left and right endpoints of $C$.
Let $u$ and $v$ be the left and right child of $x$ in $A$, respectively. We call $S \cap G_u$ the left half of $S$ and $S \cap G_v$ the right half of $S$.
Without loss of generality assume that $|S \cap G_v| \geq |S \cap G_u|$.
Let $y$ be the deepest descendant of $x$ such that $G_y$ contains the right half of $S$.
Note that it may be the case that $y = v$.
We call $y$ the charging node of $C$.
This is illustrated in \autoref{arrangement-tree}.
In the following lemma we show that the existence of $C$ implies that $c_y = \Theta(|C|)$. 
This is the charge that will pay for $C$ in the cycle basis.

\begin{lemma}\label{charge-cycle}
Let $C$ be a fundamental cycle of $T$ as described above.
It follows that $C$'s charging node $y$ has $c_y > 0$ and $y$'s left child $z$ contributes $n_z$ to its charge.
Moreover, $\frac{1}{4}(|C| - 1) \leq c_y \leq b \cdot |C|$.
\end{lemma}
\begin{proof}
Consider the two unique $v_\ell$-to-$v_r$ paths, $P_1$ and $P_2$, in $C$.
Since there are at least $b$ vertices in $G_z$ there must be edges $e_1 \in E(P_1)$ and $e_2 \in E(P_2)$ connecting $G_u$ to $G_z$.
One of these edges belongs to $T$, and the other is the fundamental edge of $C$.
The right endpoints of $e_1$ and $e_2$ must belong to long components of $G_z$ since they belong to $P_1$ and $P_2$ which extend to $v_r$.
Moreover, these long components are distinct.
For otherwise, $C$'s right endpoint would be in $G_z$, contradicting our choice of $y$.
By the existence of $P_1$ and $P_2$, these long components are merged in $G_y$.
Since $y$ is the parent of $z$ with $\ell_y < \ell_z$, we have $c_y \geq n_z$.
We also have that $c_y \leq n_y = 2n_z$.
Further, by our choice of $y$ as the deepest descendant, $n_z \leq s \leq 4n_z$.
Combining these inequalities with those of \autoref{spread-factor} yields $\frac{1}{4}(|C| - 1) \leq c_y \leq b \cdot |C|$.
\begin{figure}[!htb]
  \centering
    \includegraphics[scale=.3]{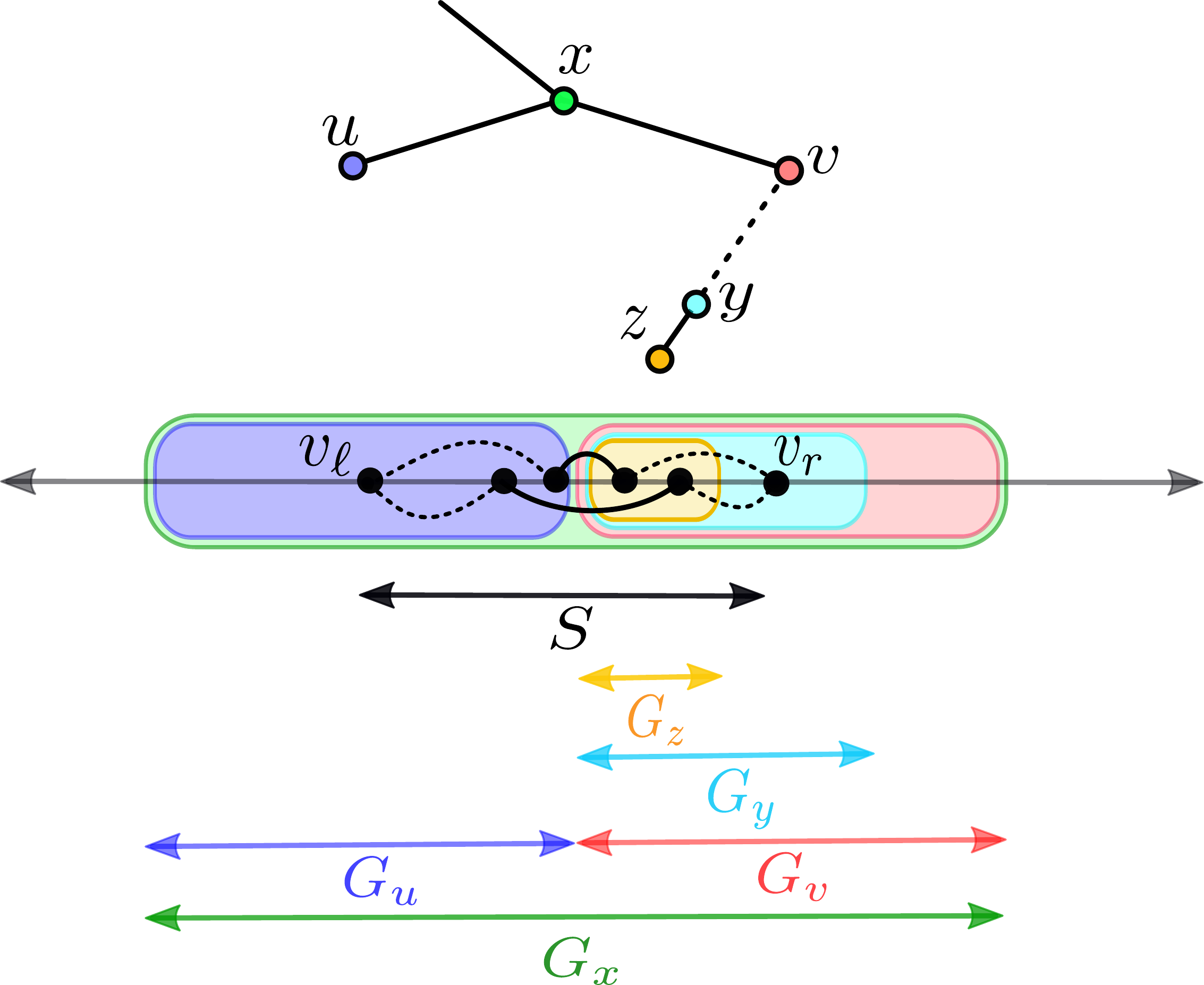}
  \caption{An illustration of the conditions of \autoref{charge-cycle}. The colored region encloses the linear arrangement of $G_x$, and the partitions represent the subgraphs induced by the descendants of $x$. The dotted lines represent the paths $P_1$ and $P_2$.
The solid lines represent the edges that induce the charge $c_y$.}
    \label{arrangement-tree}
\end{figure}
\end{proof}

We are now ready to prove the main theorem of the section.

\begin{theorem}\label{bw-stretch}
The spanning tree $T$ of $G$ produced by Algorithm~\ref{alg:st} has $\FCB(T) \leq 4b^3n$.
\end{theorem}
\begin{proof}
There are at most $\frac{1}{2}(b-1)(b-2)$ edges in $S_x$ by (\ref{edge-bound}). 
By \autoref{charge-cycle}, sum of the lengths of all of the fundamental cycles with spread at least $4b$ is at most
\begin{align*}
\sum_{y \in V(A)} \frac{1}{2}(b-1)(b-2)(4c_y + 1) &\leq n + 2(b-1)(b-2) \sum_{y \in V(A)} c_y\\
&\leq n + 2(b-1)(b-2)bn\\
&\leq 3b^3n
\end{align*}
Where the first and second inequalities come from Lemmas \ref{charge-cycle} and \ref{linear-charge}, respectively.
All fundamental cycles with spread at most $4b$ have their non-tree edges a node of $A$ of height at most $\log 4b$. 
Therefore, there are at most $n$ nodes in $A$ with $|V(G_x)| \leq 4b$ that contain a cycle.  
These contribute at most $\frac{1}{2}(b-1)(b-2)n$ to the sum of the lengths of the fundamental cycles.
In total we have
\[\FCB(T) \leq 3b^3n + b^2n \leq 4b^3n\]
as desired.
\end{proof}

\begin{corollary}\label{bw-cor}
The tree $T$ produced by our spanning tree algorithm has $\str(T) \leq 4b^3 + 2$.
\end{corollary}
\begin{proof}
According to (\ref{eqn:FCBvsStretch}), the weight of the fundamental cycle basis and the minimum stretch spanning tree are related by \[\str(T) = \frac{1}{m}(\FCB(T) - m + 2n + 2).\]  
The result follows immediately from the fact that $n \leq m \leq bn$.
\end{proof}
\section{Bounded treewidth}\label{section:treewidth}
In this section we consider simple, connected, unweighted graphs with fixed treewidth $k$.
We provide a dynamic programming approach computing a spanning tree that minimizes the total stretch over all spanning trees of $G$.
The dynamic programming table indexes partial solutions based on a localized view of the complete solution from a bag of the tree decomposition.
This is done by indexing the table with trees that correspond with weighted contracted spanning trees of $G$ that retain the stretch of the edges inside the current bag.
The approach yields a dynamic programming table whose size is polynomial in $n$ but superexponential in $k$. The goal of this section is to prove the following theorem.

\begin{theorem}
A minimum stretch spanning tree of a graph with $n$ vertices and treewidth $k$ can be computed in $O(2^{3k} k^{2k} n^{k+1})$ time.
\end{theorem}

\subsection{Spanning trees conforming to a configuration}
Let $\mathcal{T}$ be a spanning tree of $G$ and $(T, c)$ be a tuple consisting of a tree $T$ and a weight function $c$ on the edges of $T$. Fix a bag $B$ in the tree decomposition of $G$.
We say that $\mathcal{T}$ \emph{conforms} to $(T, c)$ if $\mathcal{T}$ can be transformed into $T$ by in the following way. Initialize $c(e) = 1$ for every edge in $\mathcal{T}$ and update by applying the following contractions while any of them is possible.
\begin{enumerate}
    \item If $e$ is not contained in any $(u, v)$-path where $u, v \in B$ then contract $e$.
    \item If $e=(u,v)$ where $u, v \notin B$ and $\deg_{\mathcal{T}}(v) = 2$ then contract $e$. Let $e'$ be the other edge incident to $v$. Set $c(e') \coloneqq c(e) + c(e')$.
    \item If $e=(u,v)$ where $u\in B$, $v \notin B$, and $\deg_{\mathcal{T}}(v) = 2$ then contract $e$. Let $e'$ be the other edge incident to $v$. Set $c(e') \coloneqq c(e) + c(e')$.
\end{enumerate}
$T$ is the unique minimal minor of $\mathcal{T}$ retaining the structure of the paths between vertices in $B$.
We call a tuple $(T, c)$ a \textit{configuration} of the bag $B$.  In Lemma \ref{treebound}, we will show that any spanning tree $\mathcal{T}$ conforms to a bounded number of configurations.
Our dynamic program will maintain an array of forests $\DP_i[T, c]$ indexed by a bag $B_i$ of the tree decomposition and all configurations with respect to the bag.
Each configuration at $B_i$ will describe a spanning tree $\mathcal{T}$ on $G$ that has been contracted in the way described above. 
We say a forest $F$ \textit{meets} a configuration $(T, c)$ if by following the contraction rules stated above $F$ can be transformed into $T \setminus S^A$ for some $S^A \subseteq V(T) \setminus V(B)$. We will define the subset $S^A$ in the following paragraph.
The solution stored at $\DP_i[T, c]$ will be the minimum cost forest of $G[D(B_i)]$ meeting the configuration $(T, c)$. We will describe how to calculate the cost of $F$ in the next subsection.
We will use $\DP_i[T, c]$ to refer to the total stretch of the partial solution and use $F$ to denote the partial solution that has been computed.

Let $\mathcal{T}$ be a tree built by our dynamic program conforming to $(T, c)$ and let $v_1,\dots,v_n$ be a path in $T$ such that $v_1,v_n \in V(B)$ and $v_2,\dots,v_{n-1} \in V(T) \setminus V(B)$. 
By property 3 of the tree decomposition either $v_2,\dots,v_n \in D(B)$ or $v_2,\dots,v_n \in A(B)$.
We call the vertices in $V(T) \setminus V(B)$ Steiner vertices and partition them into two sets $S^A$ and $S^B$, the Steiner vertices above the bag and the Steiner vertices below the bag.
A forest $F$ meets the configuration $(T, c)$ if it can be transformed into $T \setminus S^A$ following our contraction scheme.
The cost of $F$ is defined to be the sum $\sum_{e \in E(G)} \str_F(e)$ where $\str_F(e)$ is the stretch of $e$ in $F$ when $e$'s endpoints are in the same connected component of $F$, when $e$'s endpoints are in different connected components we set $\str_F(e)$ to be the distance between $e$'s endpoints in $T$ weighted by the cost function $c$.
Our dynamic program will process the bags of the tree decomposition in a leaf-to-root order.
Paths in $S^A$ will represent paths that will eventually be added to the complete solution by the dynamic program and paths in $S^B$ will represent paths that have already been added to the partial solution by the dynamic program.
\begin{lemma}
\label{treebound}
Let $\mathcal{T}$ be a spanning tree of $G$. There is a configuration $(T, c)$ at bag $B$ that $\mathcal{T}$ conforms to such that $|V(T)| = O(k)$.
\end{lemma}
\begin{proof}
Let $(T, c)$ be the configuration obtained by applying the contraction rules to $\mathcal{T}$.
Every vertex $v \in V(T) \setminus B$ is an internal vertex of $T$, otherwise its incident edge is not contained in a path connecting a pair of vertices from $B$ and should have been contracted. 
Further, any vertex of $V(T) \setminus B$ with degree $2$ in $T$ is adjacent to two vertices of $B$. Therefore, $T$ is a tree with at most $k+1$ leaves and $k+1$ vertices of degree $2$. It follows that $|V(T)| = O(k)$.
\end{proof}

We now describe how to populate each entry in the dynamic programming table by considering each type of bag separately.
We will prove that the forests indexed at each entry $\DP_i[T, c]$ span $D(B_i)$ and minimize the cost over all forests meeting the configuration $(T, c)$.
We will prove each case by induction using the fact that any solution stored at a leaf node is a single vertex as our base case.

\subsection{Leaf nodes}
If $B_i$ is a leaf node in the tree decomposition it contains one vertex $v$. The only configuration on $B_i$ is $(\{v\}, \emptyset)$ where $\emptyset$ is the empty function.
We initialize $F_i \coloneqq \{v\}$ and $\DP_i[\{v\}, \emptyset] \coloneqq 0$.

\subsection{Introduce nodes}
When $B_i$ is an introduce node with child $B_j$ we have $B_i = B_j \cup \{v\}$ where $v$ is the vertex being introduced to $B_i$.
Let $(T_i, c_i)$ and $(T_j, c_j)$ be configurations of $B_i$ and $B_j$.
Let $F_j$ be the partial solution stored at $\DP_j[T_j, c_j]$.
We say $(T_i, c_i)$ and $(T_j, c_j)$ are compatible with one another if $(T_i, c_i)$ can be constructed from $(T_j, c_j)$ in a way that extends $F_j$ to a partial solution $F_i$ in the following way.
If $\mathcal{T}$ is a spanning tree conforming to $(T_j, c_j)$ such that $\mathcal{T}[D(B_j)] = F_j$ we construct $F_i$ and $(T_i, c_i)$ such that $\mathcal{T}[D(B_i)] = F_i$ and $\mathcal{T}$ conforms to $(T_i, c_i)$.
We enumerate the six ways $F_j$ can be extended to $F_i$ meeting this criteria; by $N(v)$ and $I(v)$ we denote the neighbors of a vertex $v$ and the edges incident to $v$.
\begin{enumerate}

\item[I1] Let $e = (v, u) \in E(G)$ with $u \in E(G[B_i])$. Define $T_i \coloneqq T_j \cup \{e\}$ and $c_i(e) = \ell(e)$.
This extends $F_j$ to $F_i \coloneqq F_j \cup \{e\}$. 

\item[I2] Let $v$ be adjacent to some set of vertices $B_v \subseteq B$ in $G$ and let $s \in S^A_j$ such that $B_v = N(s) \cap B_j$ and $c_j(b, s) = 1$ for each $b \in B_v$.
Define $S_i^A \coloneqq S_j^A \setminus \{v\}$, $S_i^B \coloneqq S_j^B$, and $E(T_i) \coloneqq E(T_j) \cup I(v) \cap E(B_i)$ with $c_i(e) = 1$ for all $e \in I(v) \cap E(B_i)$ and $c_i(e) = c_j(e)$ for all $e \notin I(v) \cap E(B_i)$.
This extends $F_j$ to $F_i \coloneqq F_j \cup (I(v) \cap I(B_v))$.

\item[I3] Let $v$ be adjacent to some vertex $b \in B_j$ in $G$. Let $b$ be adjacent to some Steiner vertex $s \in S_j^A$ with $c_j(b, s) > 1$.
Define $T_i \coloneqq T_j \cup \{(v, b), (v, s)\}$ with $c_i(v, s) \coloneqq c_j(v, s) - 1$.
This extends $F_j$ to $F_i \coloneqq F_j \cup \{(v, b)\}$.

\item[I4] Let $s \in S_j^A$ and define $T_i \coloneqq T_j \cup \{(v, s)\}$ with $1 \leq c_i(v, s) \leq n$. This extends $F_j$ to $F_i \coloneqq F_j \cup \{v\}$.

\item[I5] Define $S_i^A \coloneqq S_j^A \cup \{s\}$ and let $b \in B_j$. Define $T_i \coloneqq T_j \cup \{(v, s), (b, s)\}$ with $1 \leq c_i(v, s) \leq n$ and $1 \leq c_i(b, s) \leq n$.
This extends $F_j$ to $F_i \coloneqq F_j \cup \{v\}$.

\item[I6] Let $s \in S_j^A$ be a Steiner vertex with $\deg(s) > 2$ and let $b \in B_j$ be adjacent to $s$ in $T_j$. We remove $(b, s)$ and introduce a new Steiner vertex $s'$ with edges $(b, s')$, $(s, s')$, and $(v, s')$. Hence $S_i^A \coloneqq S_j^A \cup \{s'\}$ and $T_i \coloneqq (T_j \setminus \{b, s\}) \cup \{(b, s'), (s, s'), (v, s')\}$ such that $c_i(b, s') + c_i(s, s') = c_j(b, s)$ and $1 \leq c_i(v, s') \leq n$.
This extends $F_j$ to $F_i \coloneqq F_j \cup \{v\}$.
\end{enumerate}

Each of these six constructions correspond to a possible way that $v$ can be connected to the complete solution constructed by the dynamic program.
See Figure \ref{introduce_fig} for an example of each case.
In I1 $v$ is directly connected to the partial solution at $\DP_i[T_i, c_i]$ via some edge in $E(B_i)$. 
In I2 and I3 $v$ can be thought of as the next vertex along the paths being built by the dynamic program. In I4, I5, and I6 $v$ is connected to the complete solution via some path that has yet to be built by the dynamic program.

\begin{figure}[!htb]
    \begin{subfigure}{.3\textwidth}
        \centering
        \includegraphics[scale=.25]{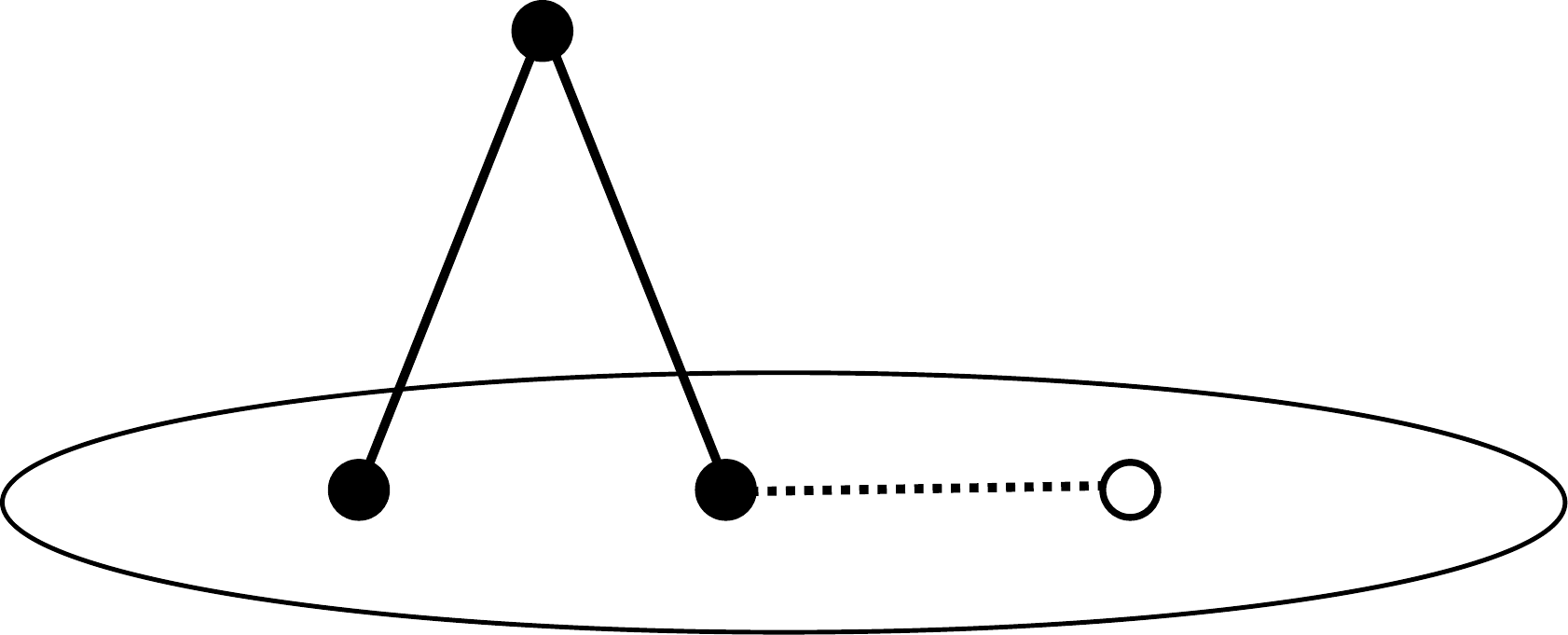}
        \renewcommand\thesubfigure{I1}
        \caption{The introduced vertex is attached to the spanning tree via an edge incident to a vertex contained in $B_i$.}
    \end{subfigure}\hspace{5mm}
    \begin{subfigure}{.3\textwidth}
        \centering
        \includegraphics[scale=.25]{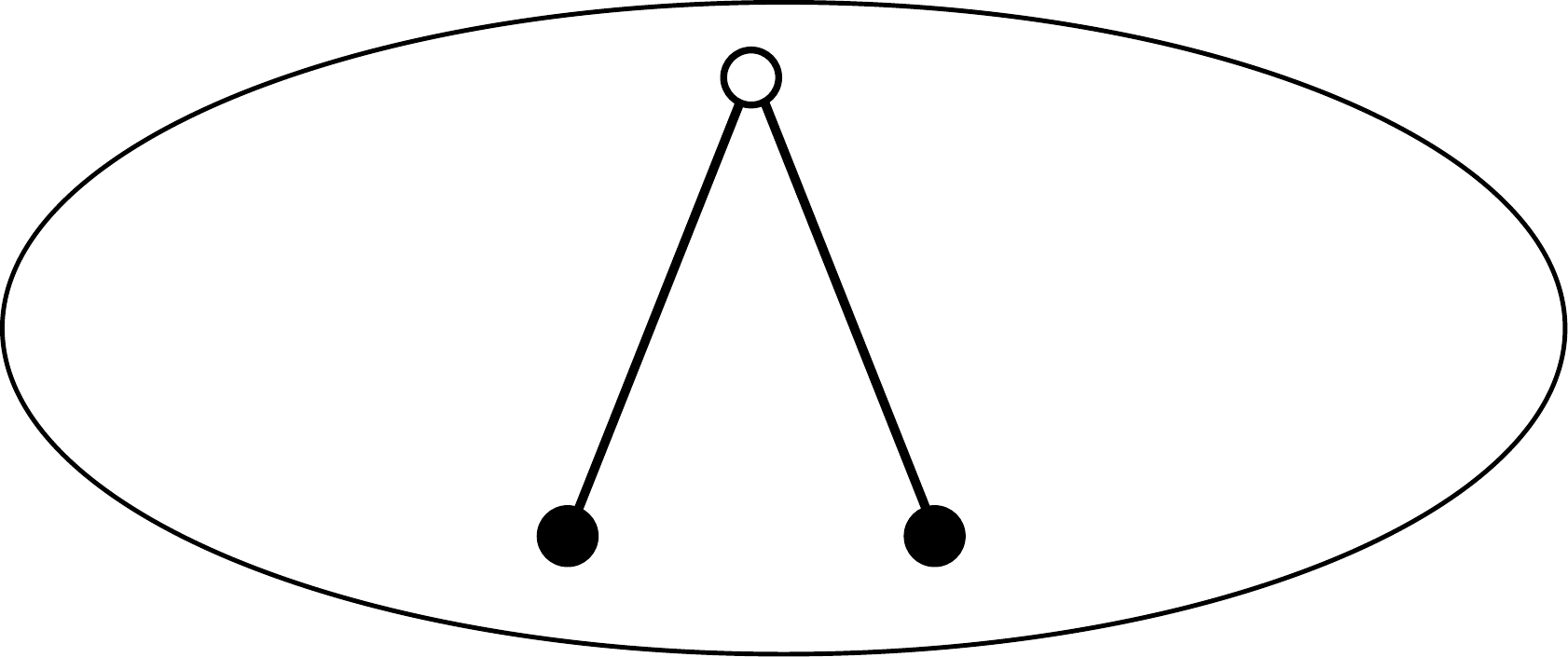}
        \renewcommand\thesubfigure{I2}
        \caption{The introduced vertex takes the role of a Steiner vertex in $S^A_i$.}
    \end{subfigure}\hspace{5mm}
    \begin{subfigure}{.3\textwidth}
        \includegraphics[scale=.25]{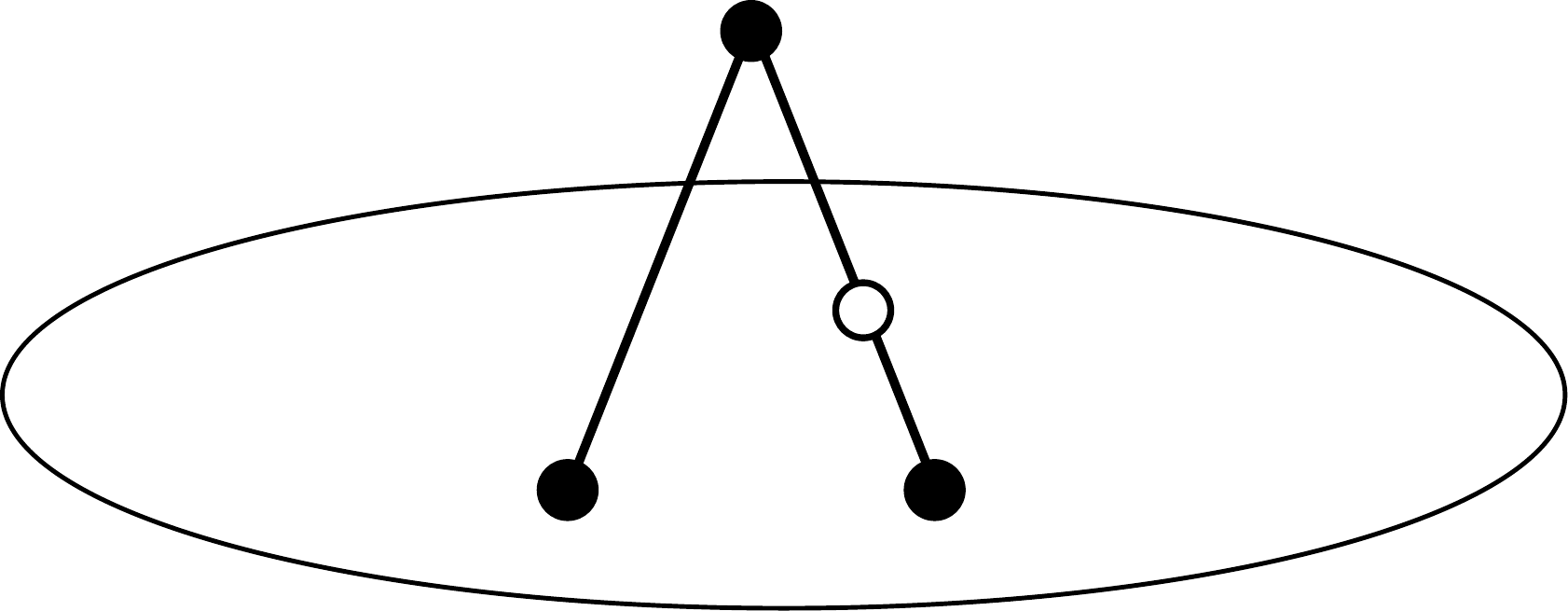}
        \renewcommand\thesubfigure{I3}
        \caption{The introduced vertex subdivides an edge between $B_i$ and $S^A_i$. The introduced vertex is the next vertex along a path being built by the dynamic program.}
    \end{subfigure}\hspace{5mm}
    \begin{subfigure}{.3\textwidth}
        \includegraphics[scale=.25]{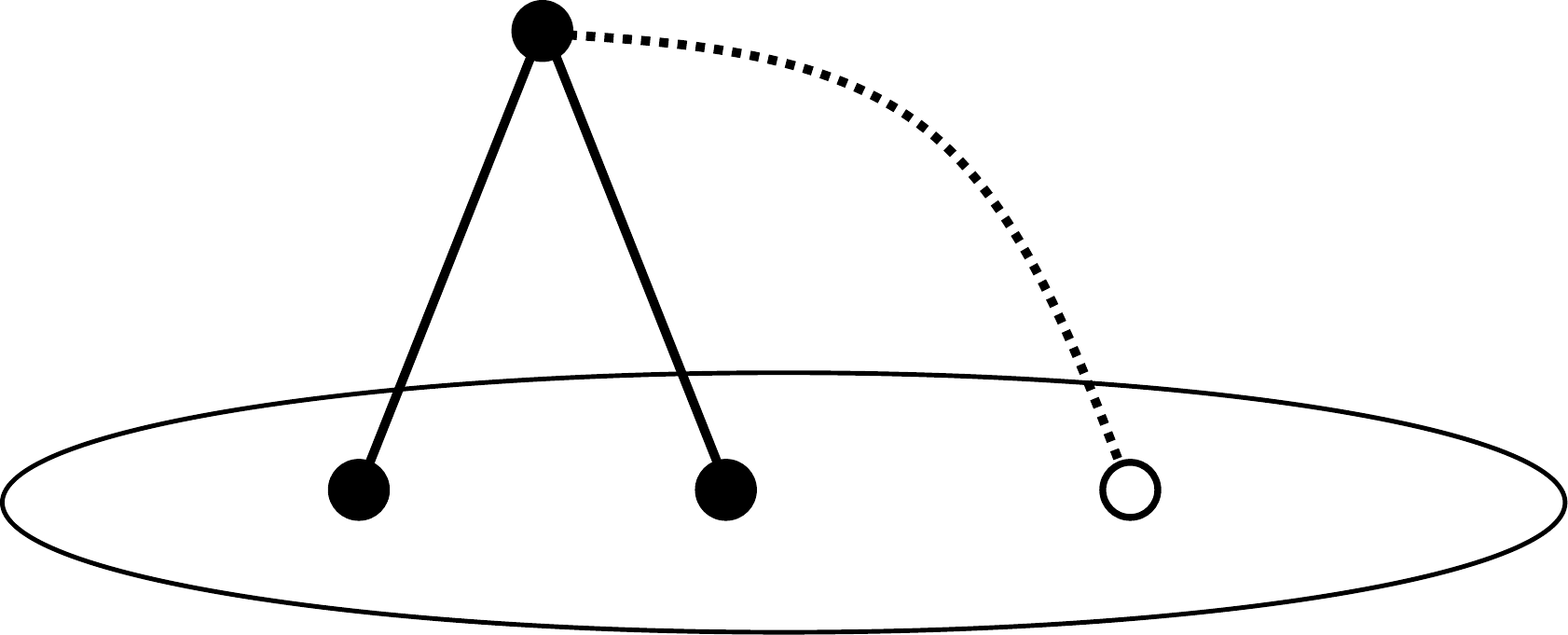}
        \renewcommand\thesubfigure{I4}
        \caption{The introduced vertex is attached to a Steiner vertex. This shows that the introduced vertex will be connected to the spanning tree along a path that the dynamic program has not yet initialized.}
    \end{subfigure}\hspace{5mm}
    \begin{subfigure}{.3\textwidth}
        \includegraphics[scale=.25]{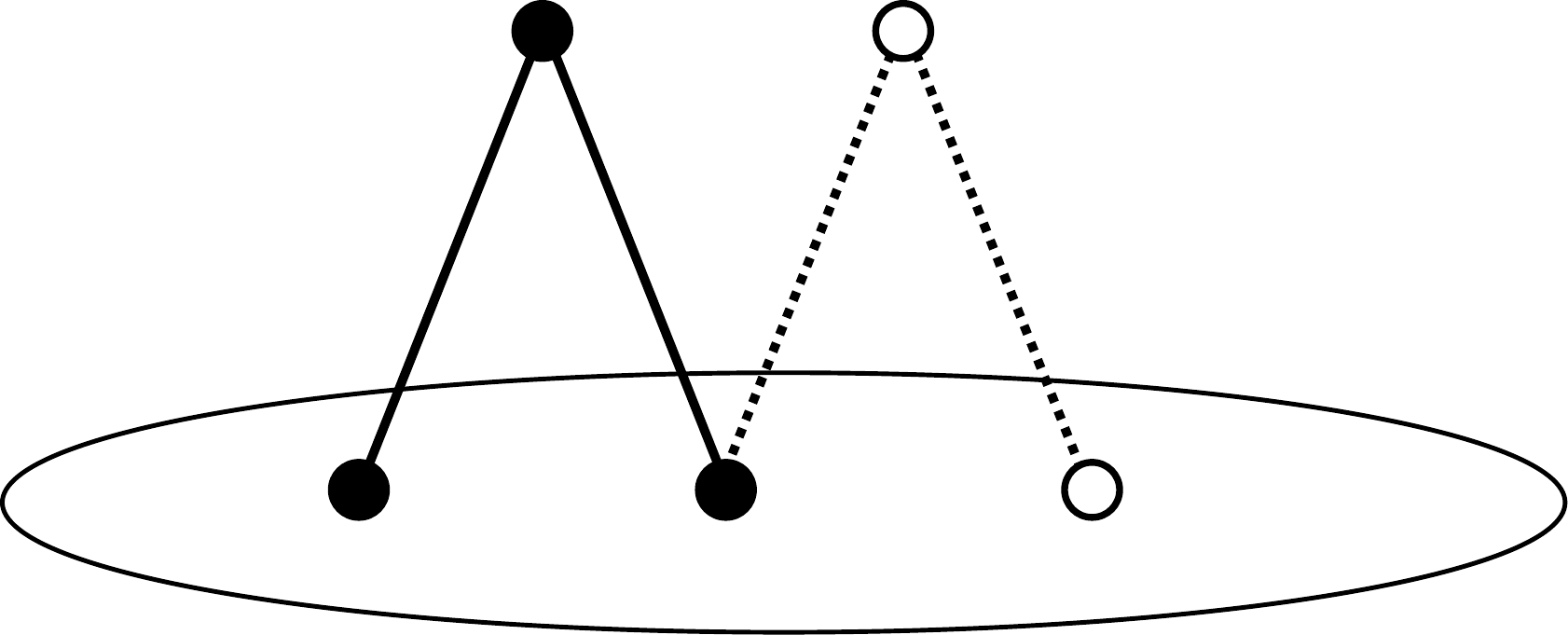}
        \renewcommand\thesubfigure{I5}
        \caption{The newly added Steiner vertex represents the intersection of two paths that have yet to be initialized by the dynamic program. The introduced vertex is connected to the spanning tree along one of these paths.}
    \end{subfigure}\hspace{5mm}
    \begin{subfigure}{.3\textwidth}
        \includegraphics[scale=.25]{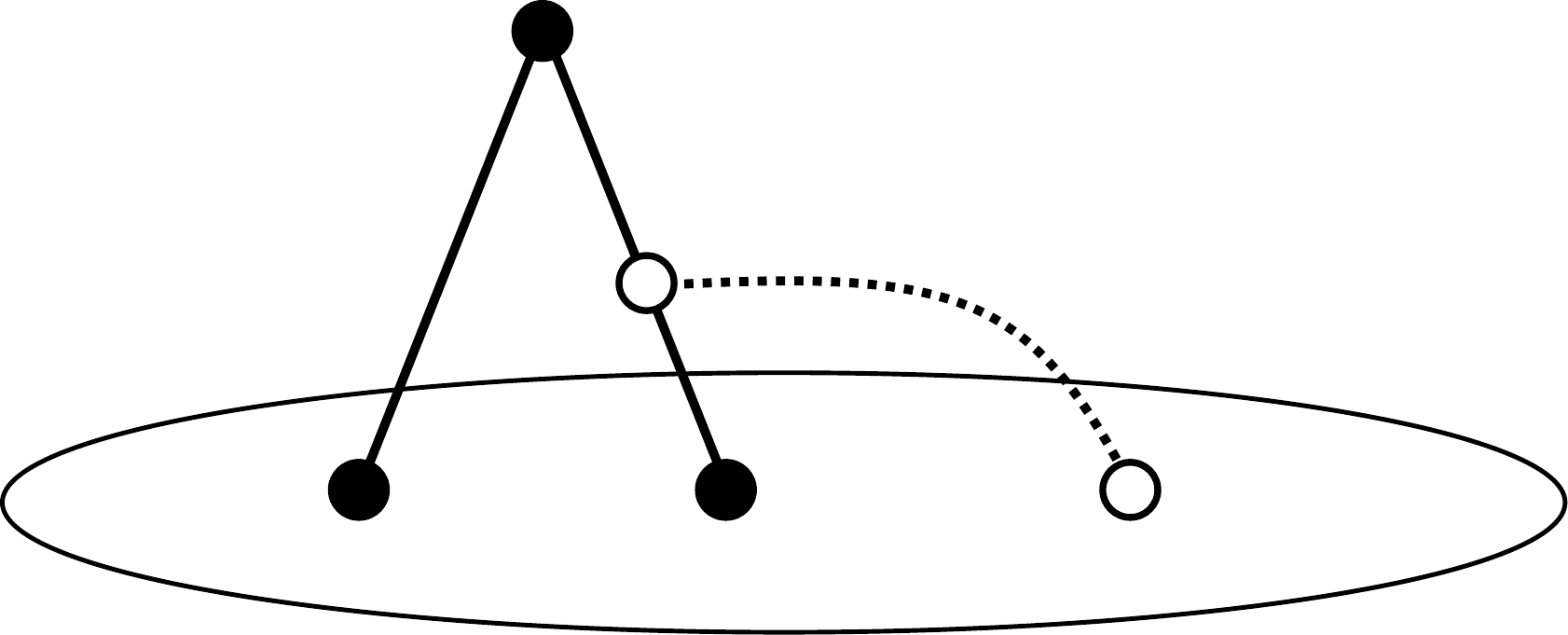}
        \renewcommand\thesubfigure{I6}
        \caption{The introduced vertex is connected to the spanning tree along a path that has not yet been initialized. The newly initialized path is attached to the spanning tree on a path that has already been initialized.}
    \end{subfigure}
    \caption{The six types of compatible configurations at an introduce node. The original tree consists of the black vertices and solid edges. The modifications are represented by the white vertices and dashed edges. The white vertex inside the circle is the vertex being introduced. The vertices enclosed in the circle are contained in the $B_i$ and the vertices above the circle are contained in $S^A_i$. }
    \label{introduce_fig}
\end{figure}

We now prove that these are the only six ways we can extend $F_j$ to $F_i$ while preserving the conformity. 
\begin{lemma}{introduceconform}
\label{introduce-conform}
Let $\mathcal{T}$ be a spanning tree of $G$ conforming to a configuration $(T_j, c_j)$ of the bag $B_j$. Let $B_i$ be the parent of $B_j$ introducing the vertex $v$. It follows that $\mathcal{T}$ conforms to a configuration $(T_i, c_i)$ of $B_i$ if and only if $(T_i, c_i)$ was constructed from $(T_j, c_j)$ via I1 through I6.
\end{lemma}
\begin{proof}
If $(T_i, c_i)$ was constructed from $(T_j, c_j)$ from one of the six methods described in the preceding subsection then either $T_i$ and $T_j$ are isomorphic (I2) and $\mathcal{T}$ conforms to $(T_i, c_i)$ or $T_i$ differs from $T_j$ by the inclusion of $v$, or the inclusion of $v$ and some Steiner vertex.
In I1, I4, I5, and I6 we have $\deg(v) = 1$ and $v$ is either adjacent to another vertex in $B_i$, a Steiner vertex with degree 2 whose second neighbor is in $B_i$, or a Steiner vertex of degree of degree at least 3. 
In each of these cases $\mathcal{T}$ conforms to $(T_i, c_i)$.
In I3 $v$ has degree two and is adjacent to a vertex in the bag and some Steiner vertex. This case is equivalent to subdividing the edge incident to the Steiner vertex to make $v$, hence the Steiner vertex still meets the conforming criteria.

Conversely, assume $\mathcal{T}$ conforms to $(T_i, c_i)$. 
If $v$ is a leaf in $T_i$ then it is connected to some other vertex in $B_i$ along some path consisting of zero or more Steiner vertices. Since $\mathcal{T}$ conforms to $(T_j, c_j)$ this path must have been contracted in $T_j$. Hence, to build $T_i$ we need to undo the contraction. This corresponds to I1, I4, I5, and I6.
If $v$ is an internal vertex in $T_i$ with $\deg(v) = 2$ with one neighbor in $S_i^A$ then $v$ must have been contracted when building $T_j$. In this case $T_i$ is built by undoing the contraction which corresponds to I3.
If $v$ is any other internal vertex in $T_i$ then it is contained in some path whose endpoints are in $B_j$. Moreover, its neighbors must also be contained in such a path otherwise they would have been contracted. It follows that $T_i$ is isomorphic to $T_j$ which corresponds to I1 where the only change is a relabeling of the vertices.
\end{proof}

The value of a subproblem at an introduce node is given by 
\begin{equation}
\label{introduce_eq}
\DP_i[T_i, c_i] = \min \left\{ \DP_j[T_j, c_j] + \sum_{e \in I(v)} \str_{T_i}(e) \right\}
\end{equation}
where the minimum is taken over all compatible configurations of $B_j$.
$F_i$ is constructed from $F_j$ and the inclusion of $v$. Since $D(B_i) = D(B_j) \cup \{v\}$ the inductive hypothesis implies that $F_i$ spans $D(B_i)$.
Finally, we show that the cost of $F_i$ is minimum over all forests meeting $(T_i, c_i)$.

\begin{lemma}
\label{introduce-stretch}
Fix a spanning tree $\mathcal{T}$ of $G$ and an introduce node $B_i$ with configuration $(T_i, c_i)$. If $\mathcal{T}$ conforms to $(T_i, c_i)$ then $\DP_i[T_i, c_i] \leq \sum_{e \in G[D(B_i)]} \str_{\mathcal{T}}(e)$.
\end{lemma}
\begin{proof}
When $B_i$ is an introduce node we have $B_i = B_j \cup \{v\}$ where $B_j$ is the child of $B_i$.
Let $(T_j, c_j)$ be a configuration of $B_j$ that is compatible with $(T_i, c_i)$.
We need to show that if $\mathcal{T}$ conforms to $(T_i, c_i)$ then $\mathcal{T}$ also conforms to $(T_j, c_j)$.
Since $(T_i, c_i)$ and $(T_j, c_j)$ are compatible $T_i$ differs from $T_j$ by at most the inclusion of $v$ and possibly a Steiner vertex $s$ adjacent to $v$. By contracting the newly added edges incident to $s$ and $v$ we see that $\mathcal{T}$ conforms to $(T_j, c_j)$.
By the inductive hypothesis we have $\DP_j[T_j, c_j] \leq \sum_{e \in G[D(B_j)]} \str_{\mathcal{T}}(e)$.
Since $D(B_i) = D(B_j) \cup \{v\}$ it follows that \[\DP_i[T_i,c_i] \leq \DP_j[T_j, c_j] + \sum_{e \in I(v)\cap B_j} \str_{\mathcal{T}}(e) \leq \sum_{e \in G[D(B_i)]} \str_{\mathcal{T}}(e).\]
\end{proof}

\subsection{Forget nodes}
When $B_i$ is a forget node with child $B_j$ we have $B_i = B_j \setminus \{v\}$ where $v$ is the vertex being forgotten in $B_i$.
Let $(T_i, c_i)$ and $(T_j, c_j)$ be configurations of $B_i$ and $B_j$.
We say $(T_i, c_i)$ and $(T_j, c_j)$ are compatible with one another if $(T_i, c_i)$ can be constructed from $(T_j, c_j)$ in the following way.
\begin{enumerate}

\item[F1.] If $v$ is a leaf construct $T_i$ by contracting the edge incident to $v$. If this edge is incident to a Steiner vertex of degree $2$ contract it as well.

\item[F2.] If $v$ is an internal vertex let $S \subseteq S^D_j$ be the set of Steiner vertices with degree 2 adjacent to $v$. Construct $T_i$ by contracting each edge $(v, s)$ for $s \in S$. Set $S^D_i \coloneqq S^D_j \cup \{v\}$ and $c_i(v, s') \coloneqq c_j(v, s) + c_j(s, b)$ where $b \in B_j$ is the other neighbor of $s$.
\end{enumerate}

\begin{figure}[!htb]
    \begin{subfigure}{.5\textwidth}
        \centering
        \includegraphics[scale=.3]{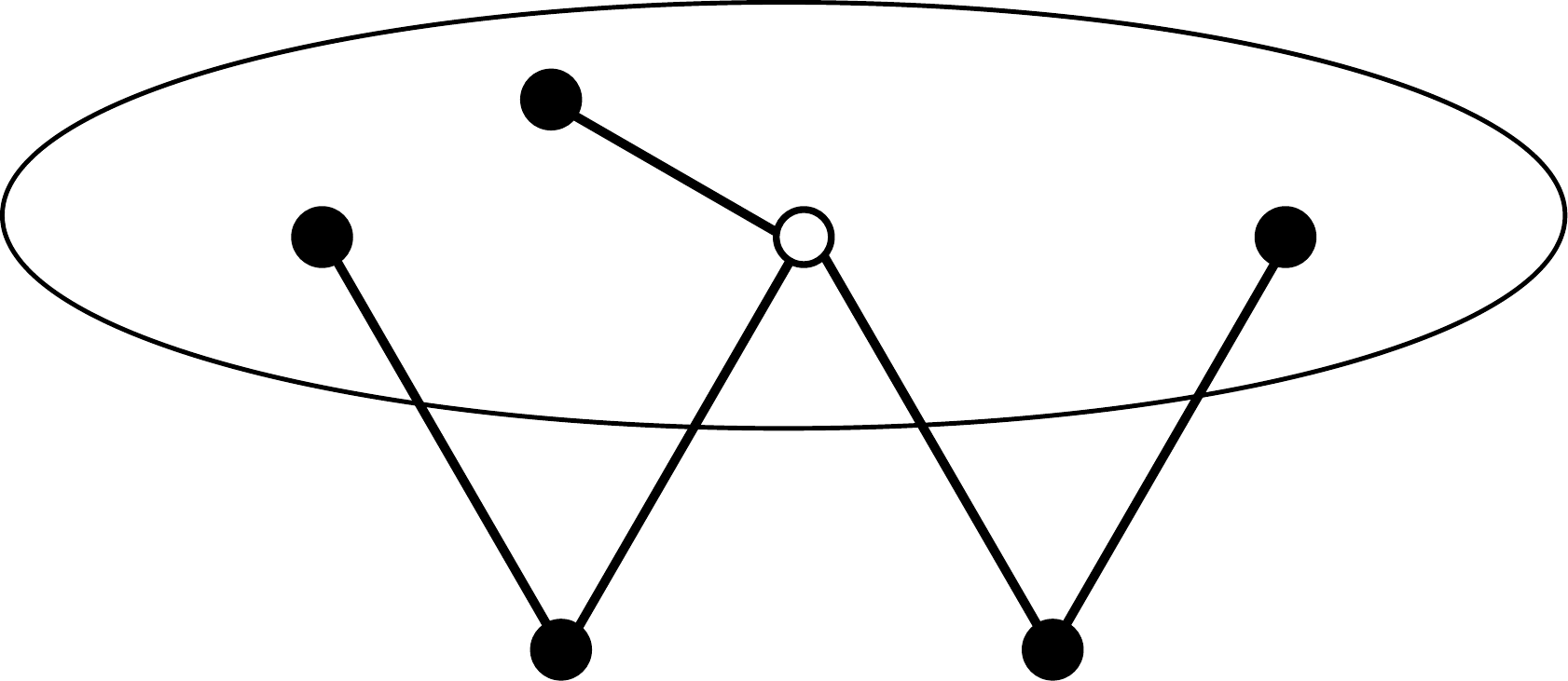}
    \end{subfigure}
    \begin{subfigure}{.5\textwidth}
        \centering
        \includegraphics[scale=.3]{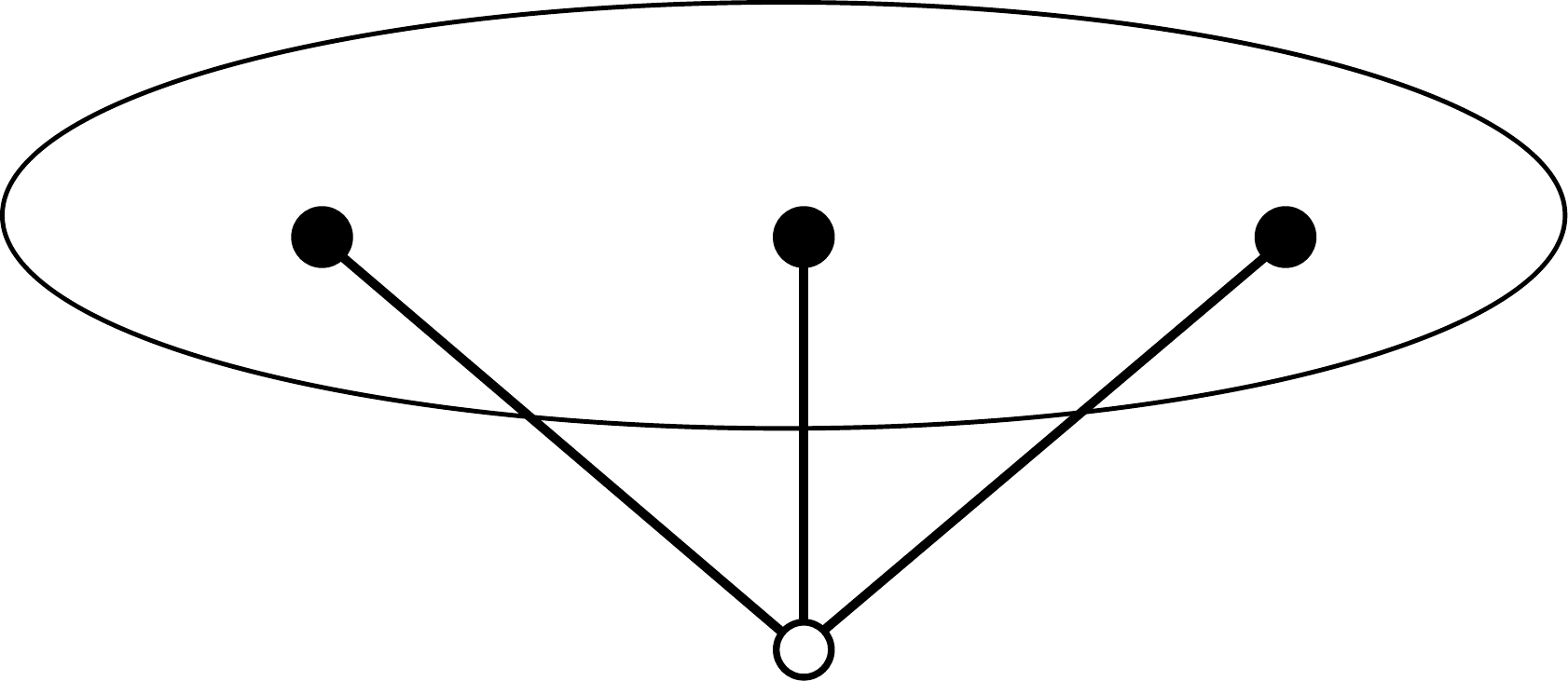}
    \end{subfigure}
    \caption{A pair of compatible configurations at a forget node. The figure on the left is the original tree, and the white vertex is the vertex being forgotten. The figure on the right is the result of the contraction.}
\end{figure}

\begin{lemma}\label{forget-conform}
Let $\mathcal{T}$ be a spanning tree of $G$ conforming to a configuration $(T_j, c_j)$ of the bag $B_j$. Let $B_i$ be the parent of $B_j$ forgetting the vertex $v$. It follows that $\mathcal{T}$ conforms to a configuration $(T_i, c_i)$ of $B_i$ if and only if $(T_i, c_i)$ was constructed from $(T_j, c_j)$ via F1 or F2.
\end{lemma}
\begin{proof}
Assume $\mathcal{T}$ conforms to $(T_i, c_i)$. Since $\mathcal{T}$ conforms to $(T_j, c_j)$ and $B_j \setminus B_i = \{v\}$ it follows that $T_i$ differs from $T_j$ by the contraction of edges incident to $v$. These edges are the edges contracted by rules F1 and F2.
Conversely, assume $(T_i, c_i)$ was constructed from $(T_j, c_j)$ by either F1 or F2.
Since F1 and F2 apply the contraction rules for conformity on the edges incident to $v$ it follows that $\mathcal{T}$ conforms to $(T_i, c_i)$.
\end{proof}

The value of a subproblem at a forget node is given by the recurrence
\begin{equation}
\label{forget_eq}
\DP_i[T_i, c_i] = \min \DP_j[T_j, c_j]
\end{equation}
where the minimum is taken over all $(T_j, c_j)$ compatible with $(T_i, c_i)$.
We set $F_i \coloneqq F_j$ where $F_j$ is the partial solution stored in the minimum $\DP_j[T_j, c_j]$.
Since $D(B_i) = D(B_j)$ it follows inductively that $F_i$ spans $D(B_i)$. We now use the inductive hypothesis to prove that $F_i$ is the minimum cost forest meeting $(T_i, c_i)$.

\begin{lemma}\label{forget-stretch}
Fix a spanning tree $\mathcal{T}$ of $G$ and a forget node $B_i$ with configuration $(T_i, c_i)$. If $\mathcal{T}$ conforms to $(T_i, c_i)$ then $\DP_i[T_i, c_i] \leq \sum_{e \in G[D(B_i)]} \str_{\mathcal{T}}(e)$.
\end{lemma}
\begin{proof}
When $B_i$ is a forget node we have $B_i = B_j \setminus \{v\}$ where $B_j$ is the child of $B_i$, hence $D(B_i) = D(B_j)$.
If $\mathcal{T}$ conforms to $(T_i, c_i)$ then $\mathcal{T}$ conforms to some configuration $(T_j, c_j)$ of $B_j$. The configuration $(T_j, c_j)$ can be found by undoing the contractions made by F1 and F2 and choosing the minimum such $\DP_j[T_j, c_j]$.
It follows that $(T_i, c_i)$ and $(T_j, c_j)$ are compatible, hence $\DP_i[T_i, c_i] = \DP_j[T_j, c_j]$.
Applying the inductive hypothesis $\DP_j[T_j, c_j] \leq \sum_{e \in G[D(B_j)]} \str_{\mathcal{T}}(e)$ proves the claim.
\end{proof}

\subsection{Join nodes}
When $B_i$ is a join node with children $B_j$ and $B_k$ we have $B_i = B_j = B_k$.
Given a configuration $(T_i, c_i)$ of $B_i$ we show how to build compatible configurations $(T_j, c_j)$ and $(T_k, c_k)$ of $B_j$ and $B_k$.
At a join node we decide which previously computed paths in the partial solutions at $B_j$ and $B_k$ to keep in the partial solution at $B_i$.

For a fixed configuration $(T, c)$ of a bag $B$ let $\mathcal{S}$ be the set of maximal, connected, induced subgraphs of $S^D$.
We \textit{invert} a tree $S \in \mathcal{S}$ by setting $S^D \coloneqq S^D \setminus S$ and $S^A \coloneqq S^A \cup S$.
If $(u, v) \in E(S)$ or $(u, v)$ has $u \in S_i^D$ and $v \in B$ we add $(u, v)$ to $E(S^A)$.
Moreover, we do not change the value of $c(u, v)$.
Inverting $S$ does not change the structure of the tree it only changes the way we interpret the Steiner vertices in $S$.

We enumerate over the subsets $\mathcal{S}'$ of $\mathcal{S}$. In one child of $B_i$ we invert $\mathcal{S}'$ and in the other we invert $\mathcal{S} \setminus \mathcal{S}'$.
For each configuration $(T_i, c_i)$ of $B_i$ and subset of trees $\mathcal{S}' \subset \mathcal{S}$ we build a compatible triplet of configurations in the following way.
Define $T_j$ to be the tree constructed by inverting $\mathcal{S}'$ in $T_i$ and $T_k$ to be the tree constructed by inverting $\mathcal{S} \setminus \mathcal{S}'$ in $T_i$.
The cost functions $c_j$ and $c_k$ are inherited from $c_i$.

\begin{figure}[!htb]
    \begin{subfigure}{.5\textwidth}
        \centering
        \includegraphics[scale=.3]{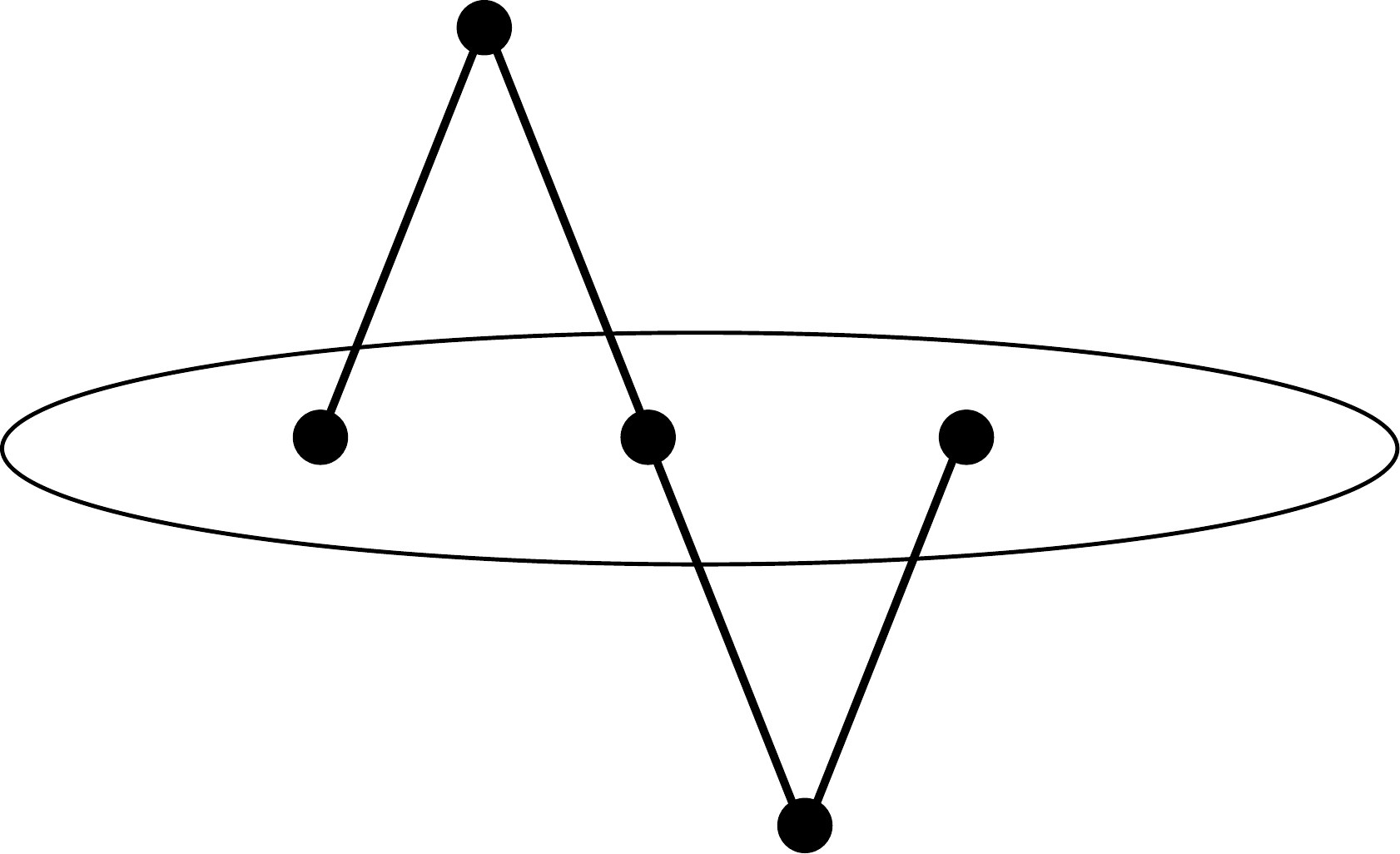}
        \renewcommand\thesubfigure{I1}
    \end{subfigure}
    \begin{subfigure}{.5\textwidth}
        \centering
        \includegraphics[scale=.3]{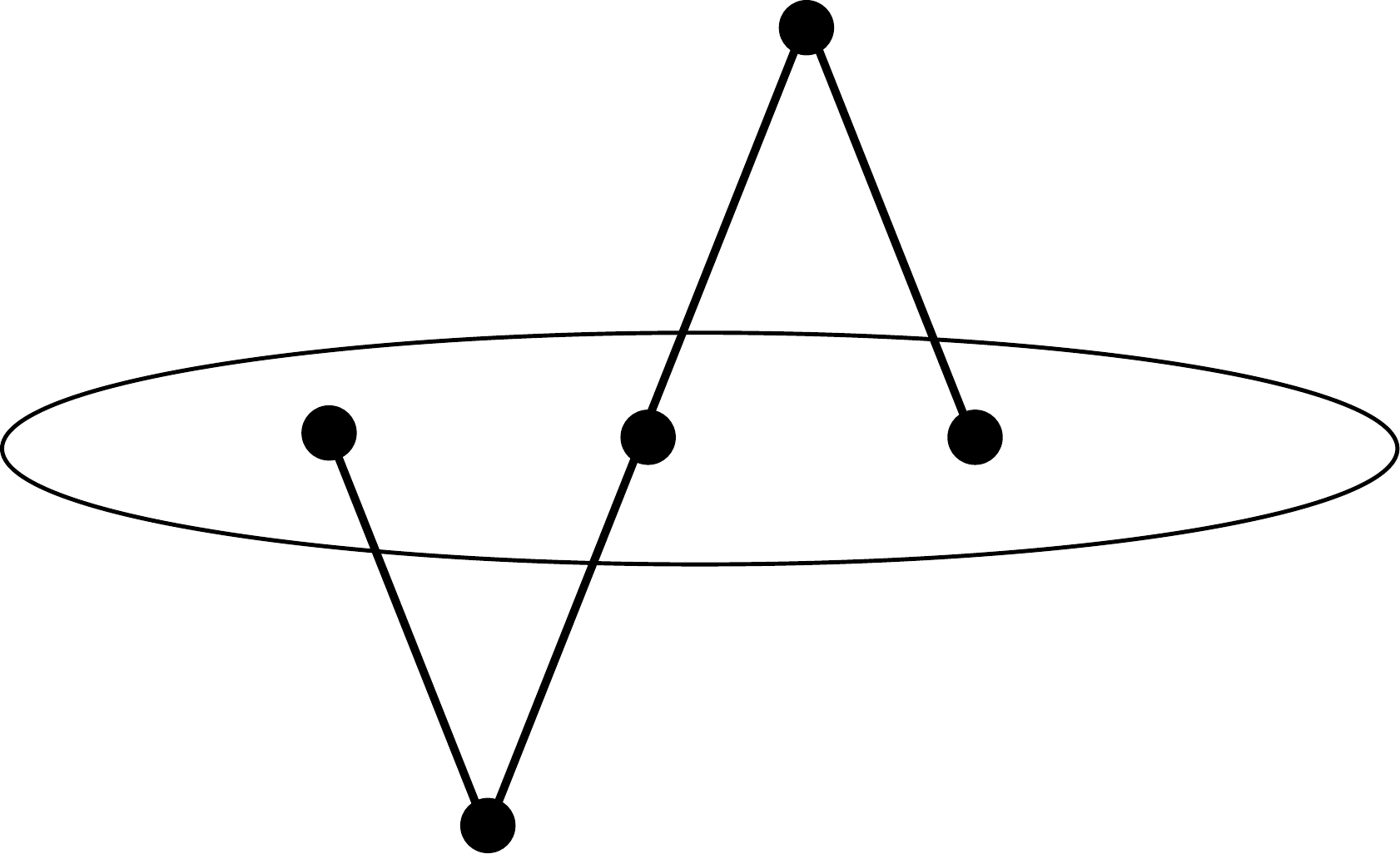}
        \renewcommand\thesubfigure{I2}
    \end{subfigure}
    \caption{A pair of compatible configurations corresponding to an inverted tree.}
\end{figure}

Fix a tree $S \in \mathcal{S}'$. The configuration $(T_j, c_j)$ is anticipating the construction of a subtree isomorphic to $S$ in order to connect the vertices in $B_j$.
Similarly, the configuration $(T_k, c_k)$ has already constructed a subtree isomorphic to $S$ connecting the vertices in $B_k$. Since $D(B_j) \setminus B_i$ and $D(B_k) \setminus B_i$ are disjoint we can safely merge the solutions to form the solution at the configuration $(T_i, c_i)$.
Hence, the stretch of the partial solution at a join node is given by 
\begin{equation}
\label{join_eq}
\DP_i[T_i, c_i] = \min \left\{ \DP_j[T_j, c_j] + \DP_k[T_k, c_k] - \sum_{e \in E(B_i)} \str_{T_i}(e)\right\}.
\end{equation}
The minimization is taken over all triplets of compatible configurations.
We subtract \[\sum_{e \in E(B_i)} \str_{T_i}(e)\] to prevent double counting the stretch of the edges in $B_i$ since $B_i = B_j = B_k$.
If $F_j$ and $F_k$ are the partial solutions at $\DP_j[T_j, c_j]$ and $\DP_k[T_k, c_k]$ then the result of the join node is $F_i \coloneqq F_j \cup F_k$.
By induction $F_j$ spans $D(B_j)$ and $F_k$ spans $D(B_k)$, hence $F_i$ spans $D(B_i)$.

\begin{lemma}{lemma}\label{join-stretch}
Fix a spanning tree of $\mathcal{T}$ of $G$ and a join node $B_i$ with configuration $(T_i, c_i)$. If $\mathcal{T}$ conforms to $(T_i, c_i)$ then $\DP_i[T_i, c_i] \leq \sum_{e \in G[D(B_i)]} \str_{\mathcal{T}}(e)$.
\end{lemma}
\begin{proof}
Let $B_i$ be a join node with children $B_j$ and $B_k$ with configurations $(T_j, c_j)$ and $(T_k, c_k)$. When $(T_i, c_i)$, $(T_j, c_j)$, and $(T_k, c_k)$ are compatible with each other the trees $T_i$, $T_j$, and $T_k$ are isomorphic since they only differ by the labeling of the Steiner vertices.
Hence, if $\mathcal{T}$ conforms to $(T_i, c_i)$ also conforms to $(T_j, c_j)$ and $(T_k, c_k)$.
By the inductive hypothesis we have $\DP_j[T_j, c_j] \leq \sum_{e \in G[D(B_j)]} \str_{\mathcal{T}}(e)$ and $\DP_k[T_k, c_k] \leq \sum_{e \in G[D(B_k)]} \str_{\mathcal{T}}(e)$.
From the equality \[\sum_{e \in G[D(B_k)]} \str_{\mathcal{T}}(e) + \sum_{e \in G[D(B_k)]} \str_{\mathcal{T}}(e) - \sum_{e \in G[B_i]} \str_{\mathcal{T}}(e) = \sum_{e \in G[D(B_i)])} \str_{\mathcal{T}}(e)\]
it follows that \[\DP_i[T_i, c_i] \leq \DP_j[T_j, c_j] + \DP_k[T_k, c_k] - \sum_{e \in G[B_i])} \str_{\mathcal{T}}(e) \leq \sum_{e \in G[D(B_i)])} \str_{\mathcal{T}}(e),\]
which proves the theorem.
\end{proof}

\subsection{Correctness}
Let $B_r$ be the root node of the tree decomposition of $G$. Without loss of generality we can assume that $B_r$ is a forget node containing one vertex $v_r$.
The only configuration on $B_r$ is $(\{v_r\}, \emptyset)$ which is a single vertex. Since every spanning tree of $G$ conforms to $(\{v_r\}, \emptyset)$ the solution indexed at $\DP_r[\{v_r\}, \emptyset]$ must be a minimum stretch spanning tree of $G$.

\subsection{Runtime analysis}
In this section we analyze the runtime of our dynamic program on a graph $G$ with treewidth $k$.
We begin by analyzing the size of the three dimensional array $\DP_i[T, c]$.
The subscript $i$ represents a bag in the nice tree decomposition of $G$.
It is known that a graph with $n$ vertices has a nice tree decomposition of width $k$ with at most $4n$ bags \cite{treewidth}.
It follows from Lemma \ref{treebound} that any tree $T$ used as an index in our array has at most $2k$ vertices.
By Cayley's formula there are at most $(2k)^{2k - 2} = O(2^{2k} k^{2k})$ such trees that will ever be built as an index by our dynamic program.
The cost function $c$ has domain $E(T)$ which has size $k$. The range of $c$ is $\{1,\dots,n\}$ since the value of the cost function is only ever incremented by one when an edge is contracted. Hence the total number of possible cost functions is $n^k$.
We conclude that the total size of our dynamic programming table is $O(2^{2k} k^{2k} n^{k+1})$.

Next we analyze the complexity of filling in the entries of our dynamic programming table. We will need to analyze introduce, forget, and join nodes as separate cases.
In each case we find the compatible configurations of the child nodes by undoing the operations described in the previous section.

At an introduce node $B_i$ we compute the value of $\DP_i[T_i, c_i]$ by undoing the six methods used to build a pair of compatible configurations. For each $v \in V(T_i) \cap B_i$ we transform $(T_i, c_i)$ into $(T_j, c_j)$ by reversing the methods described in the introduce nodes section with $v$ being treated as the vertex introduced to $B_i$. When $v$ is a leaf in $T_i$ or an internal vertex with one neighbor in $B_i$ this is done by contracting the added edges. Otherwise, we take $(T_j, c_j) \coloneqq (T_i, c_i)$. Hence, equation \ref{introduce_eq} takes the minimum over $O(k)$ compatible configurations.

When $B_i$ is a forget node forgetting a vertex $v$ there are two methods for finding compatible configurations of $(T_i, c_i)$.
In the case that $v$ was a leaf in $T_j$ we attach $v$ to each of the $O(k)$ vertices in $V(T_i) \cap B_i$ to construct each possible compatible configuration $(T_j, c_j)$. We have to consider the two cases where $v$ is adjacent to the vertex in $V(T) \cap B_j$ and where there exists one intermediate Steiner vertex of degree $2$ in between them.
In the case that $v$ was an internal vertex we consider each of the $O(k)$ Steiner vertices in $S_i^D$ that are adjacent to some vertex in $V(T_i) \cap B_i$ via some edge of cost $1$. To undo the operation we subdivide each of its incident edges with cost greater than $1$ whose endpoint is in $B_i$.
It follows that equation \ref{forget_eq} takes its minimum over $O(k)$ compatible configurations.

When $B_i$ is a join node there is a pair of compatible configurations for each of the $O(2^k)$ subsets of $S_i^D$. It follows that equation \ref{join_eq} takes the minimum over $O(2^k)$ values.
Computing the entries of $\DP_i[T, c]$ is dominated by the time it takes to compute the value at join nodes.
We have now proven the main theorem of the section.

\bibliographystyle{plain}
\bibliography{fcb}
\end{document}